\documentclass[a4paper,12pt]{article}
\usepackage{authblk}
\usepackage[dvipsnames]{xcolor}
\usepackage{latexsym,amsmath,amssymb,amsfonts,amsthm,dsfont,enumerate,natbib,bbm,threeparttable, rotating, pdflscape, verbatim, amssymb}
\usepackage[normalem]{ulem}

\usepackage[colorlinks,linkcolor=red,citecolor=blue,urlcolor=blue,breaklinks]{hyperref}
\usepackage{cleveref}

\newtheorem{theorem}{Theorem}

\newtheorem{proposition}[theorem]{Proposition}
\newtheorem{lemma}[theorem]{Lemma}

\DeclareMathOperator*{\argmin}{argmin}
\DeclareMathOperator*{\argmax}{argmax}

\RequirePackage{lineno}

\DeclareMathOperator*{\sgn}{sgn}
  
\begin{document}


\title{Confidence intervals for high-dimensional Cox models}
\author[1]{Yi Yu}
\author[2]{Jelena Bradic} 
\author[3]{Richard J. Samworth}

\affil[1]{\footnotesize{School of Mathematics, University of Bristol}}
\affil[2]{Department of Mathematics, University of California at San Diego}
\affil[3]{Statistical Laboratory, University of Cambridge}

\maketitle

\begin{abstract}
The purpose of this paper is to construct confidence intervals for the regression coefficients in high-dimensional Cox proportional hazards regression models where the number of covariates may be larger than the sample size.  Our debiased estimator construction is similar to those in \cite{ZhangZhang2014} and \cite{GeerEtal2014}, but the time-dependent covariates and censored risk sets introduce considerable additional challenges.  Our theoretical results, which provide conditions under which our confidence intervals are asymptotically valid, are supported by extensive numerical experiments.
\end{abstract}

\section{Introduction}

Over the last 45 years, the Cox proportional hazards model \citep{Cox1972} has become central to the analysis of censored survival data.  It posits that the conditional hazard rate at time $t \in \mathcal{T}$ for the survival time $\tilde{T}$ of an individual given their $p$-variate covariate vector $\boldsymbol{Z}(t)$ can be expressed as 
\begin{equation}
\label{Eq:Cox}
\lambda(t) := \lambda_0(t)\exp\bigl\{\boldsymbol{\beta}^{o\top}\boldsymbol{Z}(t)\bigr\}, 
\end{equation}
where $\boldsymbol{\beta}^o \in \mathbb{R}^p$ is an unknown vector of regression coefficients and $\lambda_0(\cdot)$ is an unknown baseline hazard function.  With $n$ individuals from a population, we assume that for each $i=1,\ldots,n$ we observe a (possibly right-censored) survival time $T_i$, an indicator $\delta_i$ of whether or not failure is observed, and the corresponding covariate processes $\{\boldsymbol{Z}_i(t): t \in \mathcal{T}\}$.


When $p < n$, the maximum partial likelihood estimator (MPLE) \citep{Cox1975} may be used to estimate $\boldsymbol{\beta}^o$.  In the classical setting where the dimension $p$ is assumed to be fixed and the sample size $n$ is allowed to diverge to infinity, and under a strong (and hard to check) condition on the weak convergence of the sample covariance processes, \cite{AndersenGill1982} derived the asymptotic normality of the MPLE using counting process arguments and Rebolledo's martingale central limit theorem.  This result may be used to provide asymptotically valid confidence intervals for components of $\boldsymbol{\beta}^o$ (or more generally, for linear combinations $\mathbf{c}^\top \boldsymbol{\beta}^o$, for some fixed $\mathbf{c} \in \mathbb{R}^p$).

Our interest in this paper lies in providing corresponding confidence intervals in the high-dimensional regime, where $p$ may be much larger than $n$.  The motivation for such methodology arises from many different application areas, but particularly in biomedicine, where Cox models are ubiquitous and data on each individual, which may arise in the form of combinations of genetic information, greyscale values for each pixel in a scan and many other types, are often plentiful.  Our construction begins with the Lasso penalised partial likelihood estimator $\widehat{\boldsymbol{\beta}}$ studied in \citet{HuangEtal2013}, which is used as an initial estimator and which is sparse.  We then seek a sparse estimator of the inverse of negative Hessian matrix, which we will refer to as a \emph{sparse precision matrix estimator}.  In \cite{ZhangZhang2014} and \cite{GeerEtal2014}, who consider similar problems in the linear and generalised linear model settings respectively, this sparse precision matrix estimator is constructed via nodewise Lasso regression \citep{MeinshausenBuhlmann2006}.  On the other hand, \cite{JavanmardMontanari2013} and \cite{JavanmardMontanari2014} derived their precision matrix estimators by minimising the trace of the product of the sample covariance matrix and the precision matrix, and the covariates are assumed to be centred. However, in the Cox model setting, the counterpart of the design matrix is a mean-shifted design matrix, where the mean is based on a set of tilting weights, and this destroys the necessary independence structure.  Instead, we adopt a modification of the CLIME estimator \citep{CaiEtal2011} as the sparse precision matrix estimator, which allows us to handle the mean subtraction.  Adjusting $\widehat{\boldsymbol{\beta}}$ by the product of our sparse precision matrix estimator and the score vector yields a debiased estimator $\widehat{\boldsymbol{b}}$, and our main theoretical result (Theorem~\ref{thm-main}) provides conditions under which $\boldsymbol{c}^\top\widehat{\boldsymbol{b}}$ is asymptotically normally distributed around $\boldsymbol{c}^\top\boldsymbol{\beta}^o$.  The desired confidence intervals can then be obtained straighforwardly.  Further very recent applications of the debiasing idea, outside the regression problem context, can found in \citet{JankovaGeer2018a} and \citet{JankovaGeer2018b}.

The formidable theoretical challenges involved in proving the asymptotic normality of $\boldsymbol{c}^\top\widehat{\boldsymbol{b}}$ arise in part from our desire to provide results in quite a general setting.  In particular, we first wish to avoid the difficult assumption on the weak convergence of sample covariance processes inherent in the martingale central limit theorem approach \citep{BradicEtal2011}.  This entails a completely different line of attack, which we believe provides new insights even in the low-dimensional setting.  Second, we wish to allow the upper limit $t_+$ of the time index set $\mathcal{T}$ to be infinite, and do not assume that each subject has a constant, positive probability of remaining in the at risk set at time $t_+$.  This is in constrast to the work of, e.g., \cite{FangEtal2016}, where the authors propose hypothesis tests based on decorrelated scores and decorrelated partial likelihood ratios (but do not provide confidence intervals for regression coefficients).  We tackle this feature of the problem by means of a novel truncation argument.  Third, our theory aims to handle settings where $p$ may be much larger than $n$; in fact, we only assume that $p = o(\exp(n^a))$, for every $a > 0$; this is sometimes called the ultrahigh dimensional setting \citep[e.g.][]{FSW2009}.

Our estimators and inference procedure are given in Section~\ref{sec:methods}, and our theoretical arguments are presented in Section~\ref{sec:theory}.  Section~\ref{sec:numerical} is devoted to extensive numerical studies of our methdology on both simulated and real data.  These reveal in particular that valid $p$-values and confidence intervals for the noise variables can be obtained with a relatively small sample size, while a larger sample size is needed for good coverage of signal variables.  Various auxiliary results are given in the Appendix.


We conclude this introduction with some notation used throughout the paper.  For any set $S$, let $|S|$ denote its cardinality.  For a vector $\boldsymbol{v} = (v_1,\ldots,v_m)^\top \in \mathbb{R}^m$, let $\|\boldsymbol{v}\|_1$, $\|\boldsymbol{v}\|$ and $\|\boldsymbol{v}\|_{\infty}$ denote its $\ell_1$, $\ell_2$ and $\ell_\infty$ norms, respectively; we also write $\boldsymbol{v}^{\otimes 2} := \boldsymbol{v}\boldsymbol{v}^\top$.  Given a set $J \subseteq \{1,\ldots,m\}$, we write $\boldsymbol{v}_J := (v_j)_{j \in J} \in \mathbb{R}^{|J|}$.  For a matrix $\boldsymbol{A} = (A_{ij})_{i,j=1}^m \in\mathbb{R}^{m\times m}$, let $\|\boldsymbol{A}\|_{\infty} := \max_{i, j = 1, \ldots, m}|A_{ij}|$ be the entrywise maximum absolute norm, and let $\|\boldsymbol{A}\|_{\mathrm{op}, \infty} := \sup_{\boldsymbol{v}\neq 0}\bigl(\|\boldsymbol{A}\boldsymbol{v}\|_{\infty}/\|\boldsymbol{v}\|_{\infty}\bigr)$ and $\|\boldsymbol{A}\|_{\mathrm{op}, 1} := \sup_{\boldsymbol{v}\neq 0}\bigl(\|\boldsymbol{A}\boldsymbol{v}\|_{1}/\|\boldsymbol{v}\|_{1}\bigr)$ denote its operator $\ell_\infty$ and operator $\ell_1$ norms respectively.  We recall in Lemma~\ref{lem:multi} in the Appendix that $\|\boldsymbol{A}\|_{\mathrm{op}, \infty}$ and $\|\boldsymbol{A}\|_{\mathrm{op}, 1}$ are, respectively, the maximum of the $\ell_1$ norms of the rows of $\boldsymbol{A}$ and the maximum of the $\ell_1$ norms of its columns.  Given two real sequences $(a_n)$ and $(b_n)$, we write $a_n \asymp b_n$ to mean $0 < \liminf_{n \rightarrow \infty} |a_n/b_n| \leq \limsup_{n \rightarrow \infty} |a_n/b_n| < \infty$.  Given a distribution function $F$, we write $\bar{F} := 1 - F$.  All probabilities and expectations are taken under the true model with baseline hazard $\lambda_0$ and regression parameter $\boldsymbol{\beta}^o$, though we suppress this in our notation.

\section{Methodology}
\label{sec:methods}

Recall that $\mathcal{T} \subseteq [0,\infty)$ denotes our time index set.  We assume that, for $i=1,\ldots,n$, there exist independent triples $\bigl(\tilde{T}_i,U_i,\{\boldsymbol{Z}_i(t): t \in \mathcal{T}\}\bigr)$, where $\tilde{T}_i$ is a non-negative random variable indicating failure time, $U_i$ is a non-negative random variable indicating a censoring time, and $\{\boldsymbol{Z}_i(t):t \in \mathcal{T}\}$ is a $p$-variate, predictable time-varying covariate process.  We further assume that $\tilde{T}_i$ and $U_i$ are conditionally independent given $\{\boldsymbol{Z}_i(t): t \in \mathcal{T}\}$.  Writing $T_i := \min(\tilde{T}_i,U_i)$ and $\delta_i := \mathbbm{1}_{\{\tilde{T}_i \leq U_i\}}$, our observations are $\bigl\{\bigl(T_i,\delta_i,\{\boldsymbol{Z}_i(t): t \in \mathcal{T}\}\bigr) \,: \, i=1,\ldots,n\bigr\}$.  We regard these observations as independent copies of a generic triple $\bigl(T,\delta,\{\boldsymbol{Z}(t): t \in \mathcal{T}\}\bigr)$.

Let $F_T$ denote the distribution function of $T$, and let $t_+ := \inf\{t \geq 0:F_T(t)=1\}$ denote the upper limit of the support of $T$.  If $t_+ < \infty$, we assume that $\mathcal{T} = [0,t_+]$; if $t_+ = \infty$, then we assume $\mathcal{T} = [0,\infty)$.  In this sense, we assume that $\mathcal{T}$ covers the entire support of the distribution of $T$, so in particular, there are no individuals in the risk set at time $t_+$.


For $i=1,\ldots,n$, define processes $\{N_i(t): t\in \mathcal{T}\}$ and $\{Y_i(t): t\in \mathcal{T}\}$ by $N_i(t) := \mathbbm{1}_{\{T_i \leq t, \delta_i = 1\}}$ and $Y_i(t) := \mathbbm{1}_{\{T_i \geq t\}}$.  We regard these as independent copies of processes $\{N(t): t\in \mathcal{T}\}$ and $\{Y(t): t\in \mathcal{T}\}$ respectively.  Let $\bar{N}(t) := n^{-1}\sum_{i=1}^n N_i(t)$.  The natural $\sigma$-field at time $t \in \mathcal{T}$ is therefore $\mathcal{F}_t := \sigma\bigl(\{(N_i(t), Y_i(t), \{\boldsymbol{Z}_i(s):s \in [0,t]\}): i = 1, \ldots, n\}\bigr)$.  The Cox model~\eqref{Eq:Cox} entails that $N_i(t)$ has predictable compensator 
\[
\Lambda_i(t,\boldsymbol{\beta}^o) := \int_0^t Y_i(s)\exp\bigl\{\boldsymbol{\beta}^{o\top}\boldsymbol{Z}_i(t)\bigr\}\lambda_0(s) \, ds
\]
with respect to the filtration $(\mathcal{F}_t:t \in \mathcal{T})$.

Define the log-partial likelihood function, divided by $n$, at $\boldsymbol{\beta} \in \mathbb{R}^p$ by 
\[
\ell(\boldsymbol{\beta}) = \ell_n(\boldsymbol{\beta}) := \frac{1}{n} \sum_{i = 1}^n \int_{\mathcal{T}} \boldsymbol{\beta}^{\top}\boldsymbol{Z}_i(s)\,dN_i(s) - \int_{\mathcal{T}}\log \biggl[\sum_{j = 1}^n Y_j(s)\exp\bigl\{\boldsymbol{\beta}^{\top}\boldsymbol{Z}_j(s)\bigr\}\biggr]\,d\bar{N}(s).
\]
Inspired by \cite{ZhangZhang2014} and \cite{GeerEtal2014}, our main object of interest is the one-step type estimator 
\begin{equation}
\label{Eq:b}
\widehat{\boldsymbol{b}} := \widehat{\boldsymbol{\beta}} + \widehat{\boldsymbol{\Theta}}\dot{\ell}(\widehat{\boldsymbol{\beta}}),
\end{equation}
where $\widehat{\boldsymbol{\beta}} = (\hat{\beta}_1,\ldots,\hat{\beta}_p)^\top$ is an initial estimator of $\boldsymbol{\beta}^o$, where $\widehat{\boldsymbol{\Theta}} = (\hat{\Theta}_{ij})_{i,j=1}^p$ is a sparse precision matrix estimator that approximates the inverse of the negative Hessian $-\ddot{\ell}(\boldsymbol{\beta}^o)$ and where $\dot{\ell}(\widehat{\boldsymbol{\beta}})$ is the score function evaluated at the initial estimator. In the rest of this section, we will elucidate the definition and rationale for our choices of $\widehat{\boldsymbol{\beta}}$ and $\widehat{\boldsymbol{\Theta}}$.  We remark that our proposals for $\widehat{\boldsymbol{\beta}}$ and $\widehat{\boldsymbol{\Theta}}$ will depend on certain tuning parameters, and this dependence is suppressed in our notation.  However, in our theoretical results, we will give explicit conditions on these tuning parameters. 

\subsection{Initial estimator}

Following \citet{HuangEtal2013}, for $\lambda > 0$, let 
\begin{equation}\label{eq-betahat}
\widehat{\boldsymbol{\beta}} = \widehat{\boldsymbol{\beta}}(\lambda) := \argmin_{\boldsymbol{\beta}\in\mathbb{R}^p} \bigl\{-\ell(\boldsymbol{\beta}) + \lambda\|\boldsymbol{\beta}\|_1\bigr\}.
\end{equation}
For $i = 1, \ldots, n$ and $t \in \mathcal{T}$, let $\tilde{w}_i(t, \boldsymbol{\beta}) := Y_i(t)\exp\{\boldsymbol{\beta}^{\top}\boldsymbol{Z}_i(t)\}$ be the $i$th weight and let
	\[
	w_i(s, \boldsymbol{\beta}) := \frac{\tilde{w}_i(s, \boldsymbol{\beta})}{\sum_{j = 1}^n \tilde{w}_j(s, \boldsymbol{\beta})}
	\]
be the $i$th normalised weight, with the convention that $0/0 := 0$.  The weighted average of the covariate processes is defined by $\bar{\boldsymbol{Z}}(s, \boldsymbol{\beta}) := \sum_{i = 1}^n \boldsymbol{Z}_i(s)w_i(s,\boldsymbol{\beta})$.  Then it follows from the subgradient conditions for optimality (Karush--Kuhn--Tucker conditions) that there exists $\hat{\boldsymbol{\tau}} = (\hat{\tau}_1,\ldots,\hat{\tau}_p)^\top$ such that
\[
0 = -\dot{\ell}(\widehat{\boldsymbol{\beta}}) + \lambda\widehat{\boldsymbol{\tau}} = -\frac{1}{n}\sum_{i=1}^n \int_{\mathcal{T}}\bigl\{\boldsymbol{Z}_i(s) - \bar{\boldsymbol{Z}}(s, \widehat{\boldsymbol{\beta}})\bigr\} \,dN_i(s) + \lambda\widehat{\boldsymbol{\tau}},
\]
where $\|\widehat{\boldsymbol{\tau}}\|_{\infty} \leq 1$ and $\hat{\tau}_j = \sgn(\hat{\beta}_j)$ if $\hat{\beta}_j \neq 0$. 


\subsection{The estimator of the precision matrix}

For $\boldsymbol{\beta} \in \mathbb{R}^p$, we have 
\[
\ddot{\ell}(\boldsymbol{\beta}) = -\sum_{i=1}^n \int_{\mathcal{T}} \bigl\{\boldsymbol{Z}_i(s) - \bar{\boldsymbol{Z}}(s,\boldsymbol{\beta})\bigr\}^{\otimes 2}w_i(s, \boldsymbol{\beta})\,d\bar{N}(s),
\]
but the presence of the weights in this integral makes it hard to analyse directly.  As a first step towards obtaining a more tractable expression, we therefore rewrite this equation as
\[
\ddot{\ell}(\boldsymbol{\beta}) = -\frac{1}{n}\sum_{i=1}^n \int_{\mathcal{T}} \bigl\{\boldsymbol{Z}_i(s) - \bar{\boldsymbol{Z}}(s,\boldsymbol{\beta})\bigr\}^{\otimes 2}\tilde{w}_i(s, \boldsymbol{\beta})\,d\widehat{\Lambda}(s,\boldsymbol{\beta}),
\]
where we define $\widehat{\Lambda}(t,\boldsymbol{\beta}) := n\int_0^t \bigl\{\sum_{j=1}^n \tilde{w}_j(s,\boldsymbol{\beta})\bigr\}^{-1} \, d\bar{N}(s)$ to be the Breslow estimator of $\int_0^t \lambda_0(s) \, ds$ \citep{Breslow1972}.  Now recall from, e.g., \citet[][p.~66]{ABGK1993} that the process $\{N(t):t \in \mathcal{T}\}$ has the Doob--Meyer decomposition
\begin{equation}
\label{Eq:DoobMeyer}
N(t) = M(t) + \int_0^t \tilde{w}(s, \boldsymbol{\beta}^{o})\lambda_0(s) \, ds,
\end{equation}
where $\{M(t): t\in \mathcal{T}\}$ is a mean-zero martingale.  This motivates us to define a population approximation to $-\ddot{\ell}(\boldsymbol{\beta}^o)$ by
\[
\boldsymbol{\Sigma} := \mathbb{E}\int_{\mathcal{T}} \{\boldsymbol{Z}(s) - \boldsymbol{\mu}(s, \boldsymbol{\beta}^o)\}^{\otimes 2}\,dN(s) = \mathbb{E}\int_0^{t_+} \bigl\{\boldsymbol{Z}(s) - \boldsymbol{\mu}(s, \boldsymbol{\beta}^o)\bigr\}^{\otimes 2}\tilde{w}(s,\boldsymbol{\beta}^o)\lambda_0(s) \, ds,
\]
where, for $t \in \mathcal{T}$ and $\boldsymbol{\beta} \in \mathbb{R}^p$,
\[
\boldsymbol{\mu}(t, \boldsymbol{\beta}) := \frac{\mathbb{E}\{\boldsymbol{Z}(t)Y(t)\exp(\boldsymbol{\beta}^{\top}\boldsymbol{Z}(t))\}}{\mathbb{E}\{Y(t)\exp(\boldsymbol{\beta}^{\top}\boldsymbol{Z}(t))\}}.
\]
Our goal in this subsection is to define an estimator of $\boldsymbol{\Sigma}^{-1}$ whose properties we can analyse.  To this end, observe that an oracle, with knowledge of $\boldsymbol{\beta}^o$, could estimate $\boldsymbol{\Sigma}$ by
\begin{align*}
\widehat{\mathcal{V}}(\boldsymbol{\beta}^o) := \frac{1}{n}\sum_{i=1}^n\int_{\mathcal{T}}\bigl\{\boldsymbol{Z}_i(s) - \bar{\boldsymbol{Z}}(s, \boldsymbol{\beta}^o)\bigr\}^{\otimes 2} \, dN_i(s) = \frac{1}{n}\sum_{i = 1}^n \delta_i\bigl\{\boldsymbol{Z}_i(T_i) - \bar{\boldsymbol{Z}}(T_i, \boldsymbol{\beta}^o)\bigr\}^{\otimes 2}.
\end{align*}
This suggests the genuine estimator 
\begin{equation}\label{eq-hatV}
\widehat{\mathcal{V}}(\widehat{\boldsymbol{\beta}}) = \frac{1}{n}\sum_{i = 1}^n \delta_i\bigl\{\boldsymbol{Z}_i(T_i) - \bar{\boldsymbol{Z}}(T_i, \widehat{\boldsymbol{\beta}})\bigr\}^{\otimes 2}.
\end{equation}
While both $-\ddot{\ell}(\widehat{\boldsymbol{\beta}})$ and $\widehat{\mathcal{V}}(\widehat{\boldsymbol{\beta}})$ can be considered as estimators of $\boldsymbol{\Sigma}$, it turns out that the latter is the much more convenient expression to study from a theoretical perspective.

As mentioned in the introduction, both \cite{ZhangZhang2014} and \cite{GeerEtal2014} employ nodewise regression to obtain a sparse precision matrix estimator $\widehat{\boldsymbol{\Theta}}$.  In those cases, the design matrices consist of independent rows, which facilitate the adoption of Lasso-type methods; in the Cox model, however, we do not have the luxury of row independence since $\widehat{\mathcal{V}}$ defined in~\eqref{eq-hatV} involves $\bar{\boldsymbol{Z}}(T_i, \widehat{\boldsymbol{\beta}})$.

As an alternative, we adapt the CLIME estimator of \citet{CaiEtal2011}, originally proposed in the context of precision matrix estimation.  
Let $\widehat{\boldsymbol{\Theta}} = (\widehat{\boldsymbol{\Theta}}_1, \ldots, \widehat{\boldsymbol{\Theta}}_p)^\top$ be defined by
	\begin{equation}\label{eq-ThetaHat}
	\widehat{\boldsymbol{\Theta}}_j \in \argmin_{\boldsymbol{b}_j \in \mathbb{R}^p}\Bigl\{\|\boldsymbol{b}_j\|_1: \, \bigl\|\widehat{\mathcal{V}}(\widehat{\boldsymbol{\beta}})\boldsymbol{b}_j - \boldsymbol{e}_j\bigr\|_{\infty} \leq \lambda_n\Bigr\},
	\end{equation}
where $\boldsymbol{e}_j^\top := (\mathbbm{1}_{\{j = l\}})_{l=1}^p \in \mathbb{R}^p$ for $j = 1, \ldots, p$.  The original proposal of \citet{CaiEtal2011} symmetrised $\widehat{\boldsymbol{\Theta}}$ by taking both the $(i,j)$th and $(j,i)$th off-diagonal entries to be the corresponding entry of $\widehat{\boldsymbol{\Theta}}$ with smaller absolute value.  In our theoretical analysis, it turned out to be convenient not to symmetrise in this way, and in practice, we found the the difference to be negligible; see Section~\ref{sec-pracissue}.

	
For $j=1,\ldots,p$, let $\dot{\ell}_j(\boldsymbol{\beta})$ denote the $j$th component of the score vector at $\boldsymbol{\beta}$, and let $\ddot{\ell}_j(\boldsymbol{\beta}) \in \mathbb{R}^p$ have $l$th component $\frac{\partial^2 \ell(\boldsymbol{\beta})}{\partial \beta_l \partial \beta_j}$.  By a Taylor expansion, for each $j = 1,\ldots,p$, there exists $\widetilde{\boldsymbol{\beta}}_j$ lying on the line segment between $\widehat{\boldsymbol{\beta}}$ and $\boldsymbol{\beta}^o$, such that
\begin{equation}
\label{Eq:ldotTaylor}
\dot{\ell}_j(\widehat{\boldsymbol{\beta}}) = \dot{\ell}_j(\boldsymbol{\beta}^o) + \ddot{\ell}_j(\widetilde{\boldsymbol{\beta}}_j)^\top (\widehat{\boldsymbol{\beta}} - \boldsymbol{\beta}^o).
\end{equation}
Now let $\boldsymbol{M}(\widetilde{\boldsymbol{\beta}}) \in \mathbb{R}^{p \times p}$ be the matrix with $j$th row $\ddot{\ell}_j(\widetilde{\boldsymbol{\beta}}_j)^\top$.  It follows that with $\widehat{\boldsymbol{b}}$ defined as in~\eqref{Eq:b}, and for any $\mathbf{c} \in \mathbb{R}^p$ with $\|\mathbf{c}\|_1 = 1$, we can write
\begin{align}\label{Eq:1}
\mathbf{c}^\top(\widehat{\boldsymbol{b}} - \boldsymbol{\beta}^o) &= \mathbf{c}^\top\bigl\{\widehat{\boldsymbol{\beta}} + \widehat{\boldsymbol{\Theta}}\dot{\ell}(\widehat{\boldsymbol{\beta}}) - \boldsymbol{\beta}^o\bigr\} \nonumber\\
&= \mathbf{c}^\top\boldsymbol{\Sigma}^{-1}\dot{\ell}(\boldsymbol{\beta}^o) + \mathbf{c}^\top\bigl(\widehat{\boldsymbol{\Theta}} - \boldsymbol{\Sigma}^{-1}\bigr)\dot{\ell}(\boldsymbol{\beta}^o) + \mathbf{c}^\top\widehat{\boldsymbol{\Theta}}\bigl\{\dot{\ell}(\widehat{\boldsymbol{\beta}}) - \dot{\ell}(\boldsymbol{\beta}^o)\bigr\} + \mathbf{c}^\top(\widehat{\boldsymbol{\beta}} - \boldsymbol{\beta}^o) \nonumber \\
&= \mathbf{c}^\top\boldsymbol{\Sigma}^{-1}\dot{\ell}(\boldsymbol{\beta}^o) + \mathbf{c}^\top(\widehat{\boldsymbol{\Theta}} - \boldsymbol{\Sigma}^{-1})\dot{\ell}(\boldsymbol{\beta}^o) + \mathbf{c}^\top\{\widehat{\boldsymbol{\Theta}}\boldsymbol{M}(\widetilde{\boldsymbol{\beta}}) + \boldsymbol{I}\}(\widehat{\boldsymbol{\beta}} - \boldsymbol{\beta}^o).
\end{align}
In Section~\ref{sec:theory} below, we will provide conditions under which, when both sides of~\eqref{Eq:1} are rescaled by $n^{1/2}$, the first, dominant term is asymptotically normal, and the second and third terms are asymptotically negligible.  This is the main step in deriving asymptotically valid confidence intervals for $\mathbf{c}^\top\boldsymbol{\beta}^o$.

\section{Theory}\label{sec:theory}

\subsection{Assumptions and main result}


Recall that our underlying processes are $n$ independent copies of the triple $\bigl(\tilde{T},U,\boldsymbol{\mathcal{Z}}\bigr)$, where $\boldsymbol{\mathcal{Z}} := \{\boldsymbol{Z}(t): t \in \mathcal{T}\}$, and that we assume $\tilde{T}$ and $U$ are conditionally independent given $\boldsymbol{\mathcal{Z}}$.  Our observations are $n$ independent copies of $\bigl(T,\delta,\{\boldsymbol{Z}(t): t \in \mathcal{T}\}\bigr)$, and we assume that the conditional hazard function of $\tilde{T}$ at time $t$ given $\boldsymbol{\mathcal{Z}}$ satisfies~\eqref{Eq:Cox}\footnote{In the terminology of, e.g., \citet[][Section~6.3]{KalbfleischPrentice2002}, this means that all time-dependent covariates are \emph{external}.} for some $\boldsymbol{\beta}^o \in \mathbb{R}^p$.  We will make use of the following assumptions:
\begin{description}
\item[\textbf{(A1)}]
\begin{description}
\item[\textbf{(a)}] The process $\{\boldsymbol{Z}(t): t\in \mathcal{T}\}$ is predictable and there exists a deterministic $K_Z > 0$ with $\sup_{t\in \mathcal{T}}\mathbb{P}\{\|\boldsymbol{Z}(t)\|_{\infty} \leq K_Z\} = 1$.
\item[\textbf{(b)}] The process $\{\boldsymbol{Z}(t): t\in \mathcal{T}\}$ is uniformly Lipschitz in the sense that there exists a deterministic $L > 0$ such that 
\[
\mathbb{P}\biggl\{\sup_{s,t \in \mathcal{T}} \|\boldsymbol{Z}(s) - \boldsymbol{Z}(t)\|_\infty \leq L|s-t|\biggr\} = 1.
\]
\end{description} 
\item[\textbf{(A2)}]
\begin{description}
\item [\textbf{(a)}]  The random variable $T$ has a bounded density $f_T$ with respect to Lebesgue measure.
\item[\textbf{(b)}] $\int_0^{t_+} t^\alpha f_T(t) \, dt < \infty$ for some $\alpha > 0$. 
\end{description}


\item[\textbf{(A3)}]

\begin{description}
\item[\textbf{(a)}] $p = p_n = o(e^{n^a})$, for every $a > 0$.
\item[\textbf{(b)}] $d_o := |\{j:\,\beta^o_j \neq 0\}|$ satisfies $d_o = o\bigl(n^{1/2}/\log^{1/2} (np)\bigr)$.
\end{description}

\item[\textbf{(A4)}] 
\begin{description}
\item [\textbf{(a)}] Writing $\mathcal{S} := \{j:\,\beta^o_j \neq 0\}$, $\mathcal{N} := \{j:\,\beta^o_j = 0\}$ and 
\[
\kappa := \inf_{\{\mathbf{v} \in \mathbb{R}^p \setminus \{0\}:\|\mathbf{v}_{\mathcal{N}}\|_1 \leq 2\|\mathbf{v}_{\mathcal{S}}\|_1\}} \frac{d_o^{1/2}\{\mathbf{v}^\top \ddot{\ell}(\boldsymbol{\beta}^o) \mathbf{v}\}^{1/2}}{\|\mathbf{v}_{\mathcal{S}}\|_1},
\]
we have that $1/\kappa = O_p(1)$.
\item [\textbf{(b)}] $\max_{j = 1, \ldots, p}\Sigma_{jj} = O(1)$ as $n \rightarrow \infty$. 
\item[\textbf{(c)}] Writing $r_j := \sum_{i=1}^p \mathbbm{1}_{\{(\boldsymbol{\Sigma}^{-1})_{ij} \neq 0\}}$ for $j=1,\ldots,p$, there exists $\delta_0 > 0$ such that  
	\[
	\|\boldsymbol{\Sigma}^{-1}\|_{\mathrm{op},1}^2\max\biggl\{\frac{d_o\log(np)}{n^{1/2}} \, , \, n^{-(1/3-\delta_0)}\biggr\} \max_{j = 1, \ldots, p}r_j  = o\biggl(\frac{1}{\log^{1/2} (np)}\biggr).
	\]

\end{description}




\end{description}

Some discussion of these assumptions is in order.  Condition \textbf{(A1)} concerns the boundedness and Lipschitz continuity of the covariate process.  It is likely that the first of these conditions could be replaced with a tail condition, at the expense of further complicating the theoretical analysis.  Indeed, in our simulations in Section~\ref{sec:numerical}, we explore settings in which $\|\boldsymbol{Z}(t)\|_\infty$ is unbounded.  Condition \textbf{(A2)} consists of two mild and interpretable conditions on the distribution of the observed failure times.  Condition \textbf{(A3)(a)} controls the rate of growth of the dimensionality as the sample size increases, and in particular allows super-polynomial growth; however, the sparsity assumption \textbf{(A3)(b)} ensures that the number of important variables (those with non-zero regression coefficient) is more tightly controlled.  Condition \textbf{(A4)(a)} is a high-level condition on the so-called compatability factor of $\ddot{\ell}(\boldsymbol{\beta}^o)$; in the presence of our other assumptions, we will see in the discussion following Lemma~\ref{Lem:lasso} that this essentially amounts to a condition on the smallest eigenvalue of $\boldsymbol{\Sigma}$.  The other parts of \textbf{(A4)} also imposes further conditions on $\boldsymbol{\Sigma}$, and, in the case of\textbf{(A4)(c)}, the way its properties interact with the sparsity level of $\boldsymbol{\beta}^o$.

The confidence intervals for the regression coefficients are constructed based on the results derived in the following theorem.
\begin{theorem}\label{thm-main}
Assume~\textbf{(A1)}-\textbf{(A4)} and let $\boldsymbol{c}\in \mathbb{R}^p$ be such that $\|\boldsymbol{c}\|_1 = 1$  and $\boldsymbol{c}^{\top}\boldsymbol{\Sigma}^{-1}\boldsymbol{c} \rightarrow \nu^2 \in (0,\infty)$.  For $\widehat{\boldsymbol{\beta}}$ in~\eqref{eq-betahat}, let $\lambda \asymp n^{-1/2}\log^{1/2} (np)$, and for $\widehat{\boldsymbol{\Theta}}$ in~\eqref{eq-ThetaHat}, let 
\[
\lambda_n \asymp \biggl\{\max\biggl(\|\boldsymbol{\Sigma}^{-1}\|_{\mathrm{op},1}\frac{d_o\log(np)}{n^{1/2}} \, , \, \|\boldsymbol{\Sigma}^{-1}\|_{\mathrm{op},1}n^{-(1/3-\delta_0)}\biggr)\biggr\}.
\]
Then for $\widehat{\boldsymbol{b}}$ defined in~\eqref{Eq:1}, we have
	\[
	n^{1/2}\boldsymbol{c}^{\top}(\widehat{\boldsymbol{b}} - \boldsymbol{\beta}^o) \stackrel{d}{\to} \mathcal{N}(0, \nu^2)
	\]
	as $n\to \infty$.  Moreover,
	\[
	n^{1/2}\boldsymbol{c}^{\top}(\widehat{\boldsymbol{b}} - \boldsymbol{\beta}^o)/(\boldsymbol{c}^{\top}\widehat{\boldsymbol{\Theta}}\boldsymbol{c})^{1/2} \stackrel{d}{\to} \mathcal{N}(0, 1).
	\]
\end{theorem}




It follows immediately from Theorem~\ref{thm-main} that for any $q \in (0,1)$, an asymptotic $(1-q)$-level confidence interval for $\boldsymbol{c}^{\top}\boldsymbol{\beta}^o$ is given by
\[
	\bigl[\boldsymbol{c}^{\top}\widehat{\boldsymbol{b}} - z_{q/2}n^{-1/2}(\boldsymbol{c}^{\top}\widehat{\boldsymbol{\Theta}}\boldsymbol{c})^{1/2}, \boldsymbol{c}^{\top}\widehat{\boldsymbol{b}} + z_{q/2}n^{-1/2}(\boldsymbol{c}^{\top}\widehat{\boldsymbol{\Theta}}\boldsymbol{c})^{1/2}\bigr],
\]
where $z_q$ is the $(1-q)$th quantile of the standard normal distribution.  In particular, for each $j=1,\ldots,p$, an asymptotic $(1-q)$-level confidence interval for $\beta_j^o$ is provided by
	\begin{equation}\label{eq-hatci}
	\bigl[\hat{b}_j - z_{q/2}n^{-1/2}(\widehat{\Theta}_{jj})^{1/2}, \hat{b}_j + z_{q/2}n^{-1/2}(\widehat{\Theta}_{jj})^{1/2}].
	\end{equation}

\subsection{Proof of Theorem~\ref{thm-main}}

The proof of Theorem~\ref{thm-main} contains three main steps: a) to provide properties of the initial estimator $\widehat{\boldsymbol{\beta}}$; b) to show the asymptotic normality of the first term in \eqref{Eq:1}; c) to show that the remainder terms in~\eqref{Eq:1} are negligible.  These steps are tackled in the following three subsections.  The final subsection completes the proof.

\subsubsection{The initial estimator}

The following lemma gives the required properties for the score function at $\boldsymbol{\beta}^o$ and the initial estimator.  The first result is proved in Lemma~3.3 of \citet{HuangEtal2013}, while the second combines Theorem~3.2 and Theorem~4.1 of the same paper.
\begin{lemma}\label{Lem:lasso}
\begin{enumerate}[(i)]
\item Assume~\textbf{(A1)(a)}.  Then for each $x > 0$,
	\begin{align*}
		\mathbb{P}\{\|\dot{\ell}(\boldsymbol{\beta}^o)\|_{\infty} > x\} \leq 2p e^{-nx^2/(8K_Z^2)}.
	\end{align*}

\item Assume~\textbf{(A1)(a)}, \textbf{(A3)(b)} and \textbf{(A4)(a)}, and take $\lambda \asymp n^{-1/2}\log^{1/2} (np)$ in~\eqref{eq-betahat}.  
Then
	\[
\|\widehat{\boldsymbol{\beta}} - \boldsymbol{\beta}^o\|_1 = O_p\biggl(\frac{d_o\log^{1/2}(np)}{n^{1/2}}\biggr).
	\]
\end{enumerate}
\end{lemma}
\textbf{Remark:} More generally, if we take a sequence $(a_n)$ diverging to infinity arbitarily slowly, and set $\lambda  \asymp n^{-1/2}\log^{1/2} (a_np)$ in~\eqref{eq-betahat}, then under the conditions of Lemma~\ref{Lem:lasso}(ii), we have $\|\widehat{\boldsymbol{\beta}} - \boldsymbol{\beta}^o\|_1 = O_p\bigl(\frac{d_o\log^{1/2}(a_np)}{n^{1/2}}\bigr)$.  In fact, if we further assume that $p = p_n \rightarrow \infty$ as $n \rightarrow \infty$, then we may take $\lambda = An^{-1/2}\log^{1/2} p$ in~\eqref{eq-betahat}, and for sufficiently large $A > 0$, conclude that $\|\widehat{\boldsymbol{\beta}} - \boldsymbol{\beta}^o\|_1 = O_p\bigl(\frac{d_o\log^{1/2} p}{n^{1/2}}\bigr)$. 
  
\bigskip

We now discuss \textbf{(A4)(a)} in greater depth.  For arbitrary finite $t^* \in \mathcal{T}$ and $M > 0$, let $C_1 := 1 + \Lambda_0(t^*)$, and let $C_2 := 2\Lambda_0(t^*)/r_*$, where $r_* := \mathbb{E}\bigl[Y(t^*)\min\{M,e^{\boldsymbol{\beta}^{o\top}\boldsymbol{Z}(t^*)}\}\bigr]$.  Further, let 
\[
\boldsymbol{\Sigma}(t^*;M) := \mathbb{E}\int_0^{t^*} \bigl\{\boldsymbol{Z}(s) - \boldsymbol{\mu}(s, \boldsymbol{\beta}^o;M)\bigr\}^{\otimes 2}Y(s) \min\{M,e^{\boldsymbol{\beta}^{o\top}\boldsymbol{Z}(t^*)}\}\lambda_0(s) \, ds,
\]
where
\[
\boldsymbol{\mu}(t, \boldsymbol{\beta}^o;M) := \frac{\mathbb{E}\bigl[\boldsymbol{Z}(t)Y(t)\min\{M,e^{\boldsymbol{\beta}^{o\top}\boldsymbol{Z}(t)}\}\bigr]}{\mathbb{E}\bigl[Y(t)\min\{M,e^{\boldsymbol{\beta}^{o\top}\boldsymbol{Z}(t)}\}\bigr]}.
\]
Write $\rho^*$ for the smallest eigenvalue of $\boldsymbol{\Sigma}(t^*;M)$, and let 
\[
t_{n,p,\epsilon} := \max\biggl\{\frac{4}{3n}\log\biggl(\frac{2.221p(p+1)}{\epsilon}\biggr) \, , \, \frac{2}{n^{1/2}}\log^{1/2}\biggl(\frac{2.221p(p+1)}{\epsilon}\biggr)\biggr\}.
\]
Then the proof of \citet[][Theorem~4.1]{HuangEtal2013} gives that for each $\epsilon \in (0,1/3)$,
\[
\mathbb{P}\biggl[\kappa < \rho^* - 36d_o K_Z^2\biggl\{\frac{2^{1/2}C_1}{n^{1/2}}\log^{1/2}\Bigl(\frac{p(p+1)}{\epsilon}\Bigr) + C_2 t_{n,p,\epsilon}^2\biggr\}\biggr] \leq 3\epsilon + e^{-nr_*^2/(8M^2)}.
\]
For fixed $t^*$ and $M$, it is natural to assume that both $\limsup_{n \rightarrow \infty} \max(C_1,C_2) < \infty$, and $\liminf_{n \rightarrow \infty} \min(\rho^*,r_*) > 0$.  In that case, under \textbf{(A3)(b)}, we have $\mathbb{P}(\kappa < \liminf_{n \rightarrow \infty} \rho^*/2) \leq 4\epsilon$ for sufficiently large $n$, so \textbf{(A4)(a)} holds. 

\subsubsection{The dominant term}

After rescaling by $n^{1/2}$, the leading term in~\eqref{Eq:1} is 
\[
n^{1/2}\boldsymbol{c}^{\top}\boldsymbol{\Sigma}^{-1}\dot{\ell}(\boldsymbol{\beta}^o) = \frac{1}{n^{1/2}}\sum_{i=1}^n \int_{\mathcal{T}}\boldsymbol{c}^{\top}\boldsymbol{\Sigma}^{-1} \bigl\{\boldsymbol{Z}_i(s) - \bar{\boldsymbol{Z}}(s, \boldsymbol{\beta}^o)\bigr\} \,dN_i(s).
\]
We will prove that its limiting distribution is Gaussian.

\begin{proposition}\label{lem-leading}
Assume \textbf{(A1)},~\textbf{(A2)}, \textbf{(A3)(a)} and \textbf{(A4)(c)}, and let $\boldsymbol{c}\in \mathbb{R}^p$ be such that $\|\boldsymbol{c}\|_1 = 1$ and $\boldsymbol{c}^{\top}\boldsymbol{\Sigma}^{-1}\boldsymbol{c} \rightarrow \nu^2 \in (0,\infty)$.  Then
\[
n^{1/2}\boldsymbol{c}^{\top}\boldsymbol{\Sigma}^{-1}\dot{\ell}(\boldsymbol{\beta}^o) \stackrel{d}{\to} \mathcal{N}(0, \nu^2),
\]
as $n\to \infty$.
\end{proposition}
\begin{proof}
Writing $M_i$ for the mean-zero martingale in the Doob--Meyer decomposition of $N_i$ (cf.~\eqref{Eq:DoobMeyer}), we have
\begin{align*}
n^{1/2}\boldsymbol{c}^{\top}\boldsymbol{\Sigma}^{-1}\dot{\ell}(\boldsymbol{\beta}^o) &= \frac{1}{n^{1/2}}\sum_{i=1}^n \int_{\mathcal{T}} \boldsymbol{c}^{\top}\boldsymbol{\Sigma}^{-1}\bigl\{\boldsymbol{Z}_i(s) - \bar{\boldsymbol{Z}}(s, \boldsymbol{\beta}^o)\bigr\} \, dN_i(s) \\
&= \frac{1}{n^{1/2}}\sum_{i=1}^n \int_{\mathcal{T}} \boldsymbol{c}^{\top}\boldsymbol{\Sigma}^{-1}\bigl\{\boldsymbol{Z}_i(s) - \bar{\boldsymbol{Z}}(s, \boldsymbol{\beta}^o)\bigr\} \, dM_i(s) \\
&= \frac{1}{n^{1/2}}\sum_{i=1}^n \int_{\mathcal{T}} \boldsymbol{c}^{\top}\boldsymbol{\Sigma}^{-1}\bigl\{\boldsymbol{Z}_i(s) - \boldsymbol{\mu}(s,\boldsymbol{\beta}^o)\bigr\} \, dM_i(s) \\
&\hspace{3cm}- \frac{1}{n^{1/2}}\sum_{i=1}^n \int_{\mathcal{T}} \boldsymbol{c}^{\top}\boldsymbol{\Sigma}^{-1}\bigl\{\bar{\boldsymbol{Z}}(s, \boldsymbol{\beta}^o) - \boldsymbol{\mu}(s,\boldsymbol{\beta}^o)\bigr\} \, dM_i(s) \\
&=: \frac{1}{n^{1/2}}\sum_{i=1}^n U_{ni} - \frac{1}{n^{1/2}}\sum_{i=1}^n V_{ni},
\end{align*}
say.  Now, for each $n \in \mathbb{N}$, we have that $U_{n1},\ldots,U_{nn}$ are independent and identically distributed, with $\mathbb{E}(U_{n1}) = 0$ and $\mathrm{Var}(U_{n1}) = \boldsymbol{c}^{\top}\boldsymbol{\Sigma}^{-1}\boldsymbol{c}$.  Moreover, for every $\epsilon > 0$,
\begin{align*}
\frac{1}{n\boldsymbol{c}^{\top}\boldsymbol{\Sigma}^{-1}\boldsymbol{c}}\sum_{i=1}^n \mathbb{E}&\bigl(U_{ni}^2\mathbbm{1}_{\{|U_{ni}| > \epsilon n^{1/2}(\boldsymbol{c}^{\top}\boldsymbol{\Sigma}^{-1}\boldsymbol{c})^{1/2}\}}\bigr) \\
&= \frac{1}{\boldsymbol{c}^{\top}\boldsymbol{\Sigma}^{-1}\boldsymbol{c}}\mathbb{E}\bigl(U_{n1}^2\mathbbm{1}_{\{|U_{n1}| > \epsilon n^{1/2}(\boldsymbol{c}^{\top}\boldsymbol{\Sigma}^{-1}\boldsymbol{c})^{1/2}\}}\bigr) \rightarrow 0
\end{align*}
as $n \rightarrow \infty$.  It follows by the Lindeberg--Feller central limit theorem \citep[e.g.][Theorem~7.2.1]{Gut2005} that 
\[
\frac{1}{n^{1/2}}\sum_{i=1}^n U_{ni} \stackrel{d}{\rightarrow} \mathcal{N}(0,\nu^2).
\]
Next, we observe that $V_{n1},\ldots,V_{nn}$ are exchangeable, with $\mathbb{E}(V_{n1}) = 0$.  Moreover, by, e.g., \citet[][pp.~74--75]{ABGK1993},
\begin{align}
\label{Eq:VarV}
\mathrm{Var}\biggl(\frac{1}{n^{1/2}}\sum_{i=1}^n V_{ni}\biggr) &= \mathrm{Var}(V_{n1}) \nonumber \\
&= \boldsymbol{c}^{\top}\boldsymbol{\Sigma}^{-1}\mathbb{E}\biggl(\int_0^{t_+} \bigl\{\bar{\boldsymbol{Z}}(s,\boldsymbol{\beta}^o) - \boldsymbol{\mu}(s,\boldsymbol{\beta}^o)\bigr\}^{\otimes 2} \tilde{w}_1(s,\boldsymbol{\beta}^o)\lambda_0(s) \, ds\biggr)\boldsymbol{\Sigma}^{-1}\boldsymbol{c}.
\end{align}
Now write $t_* := F_T^{-1}(1-n^{-1/2})$ and $\boldsymbol{\mathcal{Z}}_1 := \{\boldsymbol{Z}_1(t):t \in \mathcal{T}\}$ and let $S_1 := -\log \bar{F}_{\tilde{T}_1|\boldsymbol{\mathcal{Z}}_1}(T_1) \leq - \log \bar{F}_{T_1|\boldsymbol{\mathcal{Z}}_1}(T_1) =: Q_1$, say, where $Q_1|\boldsymbol{\mathcal{Z}}_1 \sim \mathrm{Exp}(1)$.  Then
\begin{align}
\label{Eq:InftyBound}
\mathbb{E}\biggl(&\int_0^{t_+} \bigl\|\bar{\boldsymbol{Z}}(s,\boldsymbol{\beta}^o) - \boldsymbol{\mu}(s,\boldsymbol{\beta}^o)\bigr\|_\infty^2 \tilde{w}_1(s,\boldsymbol{\beta}^o)\lambda_0(s) \, ds\biggr) \nonumber \\
&\leq \mathbb{E}\biggl\{\sup_{s \in [0,t_*)} \bigl\|\bar{\boldsymbol{Z}}(s,\boldsymbol{\beta}^o) - \boldsymbol{\mu}(s,\boldsymbol{\beta}^o)\bigr\|_\infty^2 \mathbb{E}(S_1|\boldsymbol{\mathcal{Z}}_1)\biggr\} + 4K_Z^2 \mathbb{E}\biggl(\int_{t_*}^{T_1} e^{\boldsymbol{\beta}^{o\top}\boldsymbol{Z}_1(s)} \lambda_0(s) \, ds \, \mathbbm{1}_{\{T_1 \geq t_*\}}\biggr) \nonumber \\
&\leq 2^{1/2} \biggl[\mathbb{E}\biggl\{\sup_{s \in [0,t_*)} \bigl\|\bar{\boldsymbol{Z}}(s,\boldsymbol{\beta}^o) - \boldsymbol{\mu}(s,\boldsymbol{\beta}^o)\bigr\|_\infty^4\biggr\}\biggr]^{1/2} + 4K_Z^2\mathbb{E}\bigl\{\mathbb{E}\bigl(S_1\mathbbm{1}_{\{T_1 \geq t_*\}}\bigm|\boldsymbol{\mathcal{Z}}_1\bigr)\bigr\}.
\end{align}
First note that
\begin{align}
\label{Eq:S1Bound}
  \mathbb{E}\bigl\{\mathbb{E}\bigl(S_1\mathbbm{1}_{\{T_1 \geq t_*\}}\bigm|\boldsymbol{\mathcal{Z}}_1\bigr)\bigr\} &\leq \mathbb{E}\bigl\{\mathbb{E}\bigl(Q_1\mathbbm{1}_{\{Q_1 \geq -\log \bar{F}_{T_1|\boldsymbol{\mathcal{Z}}_1}(t_*)\}}\bigm|\boldsymbol{\mathcal{Z}}_1\bigr)\bigr\} \nonumber \\
&= \mathbb{E}\bigl[\bigl\{1 - \log \bar{F}_{T_1|\boldsymbol{\mathcal{Z}}_1}(t_*)\bigr\}\bar{F}_{T_1|\boldsymbol{\mathcal{Z}}_1}(t_*)\bigr] \nonumber \\
                                                                                                                  &\leq \mathbb{E}\bar{F}_{T_1|\boldsymbol{\mathcal{Z}}_1}(t_*) - \mathbb{E}\bar{F}_{T_1|\boldsymbol{\mathcal{Z}}_1}(t_*)\log \mathbb{E}\bar{F}_{T_1|\boldsymbol{\mathcal{Z}}_1}(t_*) \nonumber \\
  &= \frac{1}{n^{1/2}} + \frac{\log n}{2n^{1/2}},                                                                                                                
\end{align}
where the second inequality follows by Jensen's inequality.

Now let $C = C(\|\boldsymbol{\beta}^o\|_1) := 1152^{1/2}K_Z\exp(2\|\boldsymbol{\beta}^o\|_1K_Z)$ and choose $\epsilon_{n,*} > 0$ such that
\[
8(p+1)n^{\max\{3,1 + 1/(2\alpha)\}}\exp\biggl(-\frac{n^{1/2}\epsilon_{n,*}^2}{2C^2}\biggr) = 1,
\]
where $\alpha > 0$ is taken from~\textbf{(A2)(b)}.  Thus, by~\textbf{(A3)(a)}, we have $\epsilon_{n,*} = o(n^{-1/4 + \delta})$ for every $\delta > 0$ and $n^{1/4}\epsilon_{n,*} \rightarrow \infty$.  Observe that 
\[
\frac{1}{n^{1/2}} = \int_{t_*}^\infty dF_T(t) \leq \frac{1}{t_*^\alpha} \int_0^\infty t^\alpha \, dF_T(t),
\]
so by~\textbf{(A2)(b)}, $t_* = O(n^{1/(2\alpha)})$.  Noting the definition of $h_0 = h_0(n,\epsilon)$ in Lemma~\ref{Lemma:Zbarmu} in the Appendix, we choose $n_0 \in \mathbb{N}$ large enough that the following conditions hold for $n \geq n_0$:
\begin{enumerate}
\item $\epsilon_{n,*} \leq 2K_Z$
\item $n\bar{F}_T(t_*) - n^{1/2}(\log n)\bar{F}_T^{1/2}(t_*) = n^{1/2} - n^{1/4}\log n \geq n^{1/2}/2$
\item $\epsilon_{n,*} - \frac{6K_Ze^{\|\boldsymbol{\beta}^o\|_1K_Z}}{n^{1/2}} \geq \epsilon_{n,*}/2$
\item $1 + t_*/h_0(n,\epsilon_{n,*}) \leq n^{1 + 1/(2\alpha)}$.
\end{enumerate}
It follows that $M_0^* := M_0(n,\epsilon_{n,*})$, defined in Lemma~\ref{Lemma:Zbarmu}, satisfies $M_0^* \leq 2n^{\max\{3,1 + 1/(2\alpha)\}}$ for $n \geq n_0$.  Write 
\[
g(n) := \frac{1}{2n} + e^{-(\log^2 n)/2} + 4n^{\max\{3,1 + 1/(2\alpha)\}} \exp\biggl\{-\frac{3n^{1/2}}{28\exp(4\|\boldsymbol{\beta}^o\|_1K_Z)}\biggr\}.
\]
Then, by Lemma~\ref{Lemma:Zbarmu}, for $n \geq n_0$,
\begin{align}
\label{Eq:InftyBound2}
\mathbb{E}\biggl\{\sup_{s \in [0,t_*)} &\bigl\|\bar{\boldsymbol{Z}}(s,\boldsymbol{\beta}^o) - \boldsymbol{\mu}(s,\boldsymbol{\beta}^o)\bigr\|_\infty^4\biggr\} = \int_0^{16K_Z^4} \mathbb{P}\biggl(\sup_{s \in [0,t_*)} \bigl\|\bar{\boldsymbol{Z}}(s,\boldsymbol{\beta}^o) - \boldsymbol{\mu}(s,\boldsymbol{\beta}^o)\bigr\|_\infty^4 > \delta\biggr) \, d\delta \nonumber \\
&\leq \epsilon_{n,*}^4 + 4\int_{\epsilon_{n,*}}^{2K_Z} \epsilon^3 \mathbb{P}\biggl(\sup_{s \in [0,t_*)} \bigl\|\bar{\boldsymbol{Z}}(s,\boldsymbol{\beta}^o) - \boldsymbol{\mu}(s,\boldsymbol{\beta}^o)\bigr\|_\infty > \epsilon\biggr) \, d\epsilon \nonumber \\
&\leq \epsilon_{n,*}^4 + 8(p+1)M_0^*\int_{\epsilon_{n,*}}^{\infty} \epsilon^3 \exp\biggl(-\frac{n^{1/2}\epsilon^2}{2C^2}\biggr) \, d\epsilon + 16K_Z^4 g(n) \nonumber \\
&= \epsilon_{n,*}^4 + \frac{8(p+1)M_0^*C^4}{n} \int_{n^{1/4}\epsilon_{n,*}/C}^{\infty} t^3 e^{-t^2/2} \, dt + 16K_Z^4 g(n) \nonumber \\
&\leq \epsilon_{n,*}^4 + \frac{2C^4\bigl\{\log\bigl(16(p+1)n^{\max\{3,1 + 1/(2\alpha)\}}\bigr)+1\bigr\}}{n} + 16K_Z^4 g(n) = o(n^{-(1-\delta)}),
\end{align}
for every $\delta > 0$.  From~\eqref{Eq:VarV},~\eqref{Eq:InftyBound},~\eqref{Eq:S1Bound} and~\eqref{Eq:InftyBound2} and \textbf{(A4)(c)}, we deduce that
\[
\mathrm{Var}\biggl(\frac{1}{n^{1/2}}\sum_{i=1}^n V_{ni}\biggr) \leq \|\boldsymbol{\Sigma}^{-1}\|_{\mathrm{op},1}^2\mathbb{E}\biggl(\int_0^{t_+} \bigl\|\bar{\boldsymbol{Z}}(s,\boldsymbol{\beta}^o) - \boldsymbol{\mu}(s,\boldsymbol{\beta}^o)\bigr\|_\infty^2 \tilde{w}_1(s,\boldsymbol{\beta}^o)\lambda_0(s) \, ds\biggr) \rightarrow 0,
\]
as required.
\end{proof}

\subsubsection{The remainder terms}

The two remainder terms in~\eqref{Eq:1} are controlled in Propositions~\ref{prop-tb1} and~\ref{prop-tb2} below respectively.

\begin{proposition}\label{prop-tb1}
Assume~\textbf{(A1)},~\textbf{(A2)(a)},~\textbf{(A3)(b)},~\textbf{(A4)(a)} and~\textbf{(A4)(c)}.  For $\widehat{\boldsymbol{\beta}}$ in~\eqref{eq-betahat}, let $\lambda \asymp n^{-1/2}\log^{1/2} (np)$, and for $\widehat{\boldsymbol{\Theta}}$ in~\eqref{eq-ThetaHat}, let 
\[
\lambda_n \asymp \biggl\{\max\biggl(\|\boldsymbol{\Sigma}^{-1}\|_{\mathrm{op},1}\frac{d_o\log (np)}{n^{1/2}} \, , \, \|\boldsymbol{\Sigma}^{-1}\|_{\mathrm{op},1}n^{-(1/3-\delta_0)}\biggr)\biggr\}.
\]
Then for $\boldsymbol{c}\in \mathbb{R}^p$ with $\|\boldsymbol{c}\|_1 = 1$, we have
	\[
	\boldsymbol{c}^{\top}\bigl(\widehat{\boldsymbol{\Theta}} - \boldsymbol{\Sigma}^{-1}\bigr)\dot{\ell}(\boldsymbol{\beta}^o) = o_p(n^{-1/2}).
	\]
\end{proposition}
\begin{proof}
Define the event $\mathcal{A} := \{\|\boldsymbol{\Sigma}^{-1}\widehat{\mathcal{V}}(\widehat{\boldsymbol{\beta}}) - \boldsymbol{I}\|_{\infty} \leq \lambda_n\} = \{\|\widehat{\mathcal{V}}(\widehat{\boldsymbol{\beta}})\boldsymbol{\Sigma}^{-1} - \boldsymbol{I}\|_{\infty} \leq \lambda_n\}$.  Then by construction of $\widehat{\boldsymbol{\Theta}} = (\widehat{\boldsymbol{\Theta}}_1,\ldots,\widehat{\boldsymbol{\Theta}}_p)^\top$ as an estimator of $\boldsymbol{\Sigma}^{-1} = \bigl((\boldsymbol{\Sigma}^{-1})_1,\ldots,(\boldsymbol{\Sigma}^{-1})_p\bigr)^\top$, on the event $\mathcal{A}$, we have
	\[
	\|\widehat{\boldsymbol{\Theta}}_j\|_1 \leq \|(\boldsymbol{\Sigma}^{-1})_j\|_1, \quad j=1,\ldots,p,
	\]
so in particular, $\|\widehat{\boldsymbol{\Theta}}\|_{\mathrm{op},\infty} \leq \|\boldsymbol{\Sigma}^{-1}\|_{\mathrm{op},\infty}$,	and
	\[
	\bigl\|\widehat{\boldsymbol{\Theta}}\widehat{\mathcal{V}}(\widehat{\boldsymbol{\beta}}) - \boldsymbol{I}\bigr\|_{\infty} \leq \lambda_n.
	\]
	Hence, using the fact that $\boldsymbol{\Sigma}$ and $\widehat{\mathcal{V}}(\widehat{\boldsymbol{\beta}})$ are symmetric, on the event $\mathcal{A}$,
	\begin{align}
 \|\widehat{\boldsymbol{\Theta}} - \boldsymbol{\Sigma}^{-1}\|_{\infty} &= \bigl\|\widehat{\boldsymbol{\Theta}}\bigl(\boldsymbol{I} - \widehat{\mathcal{V}}(\widehat{\boldsymbol{\beta}})\boldsymbol{\Sigma}^{-1}\bigr) + \bigl(\widehat{\boldsymbol{\Theta}}\widehat{\mathcal{V}}(\widehat{\boldsymbol{\beta}}) - \boldsymbol{I}\bigr)\boldsymbol{\Sigma}^{-1}\bigr\|_{\infty} \nonumber\\
	&\leq  \|\widehat{\boldsymbol{\Theta}}\|_{\mathrm{op},\infty}\|\widehat{\mathcal{V}}(\widehat{\boldsymbol{\beta}})\boldsymbol{\Sigma}^{-1} - \boldsymbol{I}\|_{\infty} +  \|\boldsymbol{\Sigma}^{-1}\|_{\mathrm{op},1}\|\widehat{\boldsymbol{\Theta}}\widehat{\mathcal{V}}(\widehat{\boldsymbol{\beta}}) - \boldsymbol{I}\|_\infty \leq 2\lambda_n\|\boldsymbol{\Sigma}^{-1}\|_{\mathrm{op},1}. \label{eq-prop2-im}
	\end{align}
Hence, from Lemma~\ref{lem-aclime-lem-7}, on the event $\mathcal{A}$,
\begin{align*}
	|\boldsymbol{c}^{\top}(\widehat{\boldsymbol{\Theta}} - \boldsymbol{\Sigma}^{-1})\dot{\ell}(\boldsymbol{\beta}^o)| &\leq \|\widehat{\boldsymbol{\Theta}} - \boldsymbol{\Sigma}^{-1}\|_{\mathrm{op},\infty}\|\dot{\ell}(\boldsymbol{\beta}^o)\|_{\infty} \\
	 &\leq 24\lambda_n\|\boldsymbol{\Sigma}^{-1}\|_{\mathrm{op},1}\|\dot{\ell}(\boldsymbol{\beta}^o)\|_{\infty}\max_{j=1,\ldots,p} r_j
\end{align*}
The conclusion therefore follows from Lemmas~\ref{Lem:lasso}(i) and~\ref{lem-aclime-lem1}, together with~\textbf{(A4)(c)}. 
\end{proof}
Recall the definition of the matrix $\boldsymbol{M}(\widetilde{\boldsymbol{\beta}})$, which is defined just after~\eqref{Eq:ldotTaylor}, and which appears in~\eqref{Eq:1}.
\begin{proposition}\label{prop-tb2}
  Assume~\textbf{(A1)},~\textbf{(A2)(a)},~\textbf{(A3)(b)} and \textbf{(A4)}.  For $\widehat{\boldsymbol{\beta}}$ in \eqref{eq-betahat}, let $\lambda \asymp n^{-1/2}\log^{1/2} (np)$, and for $\widehat{\boldsymbol{\Theta}}$ in~\eqref{eq-ThetaHat}, let 
\[
\lambda_n \asymp \biggl\{\max\biggl(\|\boldsymbol{\Sigma}^{-1}\|_{\mathrm{op},1}\frac{d_o\log (np)}{n^{1/2}} \, , \, \|\boldsymbol{\Sigma}^{-1}\|_{\mathrm{op},1}n^{-(1/3-\delta_0)}\biggr)\biggr\}.
\]
Then for $\boldsymbol{c}\in \mathbb{R}^p$ with $\|\boldsymbol{c}\|_1 = 1$, we have
	\[
	\boldsymbol{c}^{\top}(\widehat{\boldsymbol{\Theta}}\boldsymbol{M}(\widetilde{\boldsymbol{\beta}}) + \boldsymbol{I})(\widehat{\boldsymbol{\beta}} - \boldsymbol{\beta}^o) = o_p(n^{-1/2}).
	\]
\end{proposition}
\begin{proof}
	Note that
	\begin{align}
\label{Eq:Start}
 \|\widehat{\boldsymbol{\Theta}} &\boldsymbol{M}(\widetilde{\boldsymbol{\beta}}) + \boldsymbol{I}\|_{\infty} \nonumber \\
	&\leq \|\widehat{\boldsymbol{\Theta}} - \boldsymbol{\Sigma}^{-1}\|_{\mathrm{op},\infty}\|\boldsymbol{M}(\widetilde{\boldsymbol{\beta}}) + \boldsymbol{\Sigma}\|_{\infty} + \|\boldsymbol{\Sigma}^{-1}\|_{\mathrm{op},1}\|\boldsymbol{M}(\widetilde{\boldsymbol{\beta}}) + \boldsymbol{\Sigma}\|_{\infty} + \|\widehat{\boldsymbol{\Theta}} - \boldsymbol{\Sigma}^{-1}\|_{\mathrm{op},\infty}\|\boldsymbol{\Sigma}\|_{\infty}.
	\end{align}
For $j=1,\ldots,p$, let $\ddot{\ell}_j(\boldsymbol{\beta}^o)$ denote the $j$th column of $\ddot{\ell}(\boldsymbol{\beta}^o)$.  Then
\begin{align}
\label{Eq:Mj}
\|&\ddot{\ell}_j(\widetilde{\boldsymbol{\beta}}_j) - \ddot{\ell}_j(\boldsymbol{\beta}^o)\|_\infty \nonumber \\
&\leq \biggl\| \sum_{i=1}^n \int_{\mathcal{T}} \biggl[\bigl\{\boldsymbol{Z}_i(s) - \bar{\boldsymbol{Z}}(s,\widetilde{\boldsymbol{\beta}}_j)\bigr\}^{\otimes 2}w_i(s,\widetilde{\boldsymbol{\beta}}_j) - \bigl\{\boldsymbol{Z}_i(s) - \bar{\boldsymbol{Z}}(s,\boldsymbol{\beta}^o)\bigr\}^{\otimes 2}w_i(s,\boldsymbol{\beta}^o)\biggr] \, d\bar{N}(s)\biggr\|_\infty \nonumber \\
&\leq \sup_{s \in \mathcal{T}} \biggl\|\sum_{i=1}^n\biggl[\bigl\{\boldsymbol{Z}_i(s) - \bar{\boldsymbol{Z}}(s,\widetilde{\boldsymbol{\beta}}_j)\bigr\}^{\otimes 2}w_i(s,\widetilde{\boldsymbol{\beta}}_j) - \bigl\{\boldsymbol{Z}_i(s) - \bar{\boldsymbol{Z}}(s,\boldsymbol{\beta}^o)\bigr\}^{\otimes 2}w_i(s,\boldsymbol{\beta}^o)\biggr]\biggr\|_\infty \nonumber \\
&\leq 4K_Z \sup_{s \in \mathcal{T}}\bigl\|\bar{\boldsymbol{Z}}(s,\widetilde{\boldsymbol{\beta}}_j) - \bar{\boldsymbol{Z}}(s,\boldsymbol{\beta}^o)\bigr\|_\infty + 4K_Z^2 \sup_{s \in \mathcal{T}} \sum_{i=1}^n |w_i(s,\widetilde{\boldsymbol{\beta}}_j) - w_i(s,\boldsymbol{\beta}^o)|.
\end{align}
But, for any $s \in \mathcal{T}$,
\begin{align}
\label{Eq:ZbarDiff2}
\bigl\|\bar{\boldsymbol{Z}}(s,\widetilde{\boldsymbol{\beta}}_j) - \bar{\boldsymbol{Z}}(s,\boldsymbol{\beta}^o)\bigr\|_\infty &= \biggl\|\frac{\sum_{i=1}^n \boldsymbol{Z}_i(s)Y_i(s)e^{\widetilde{\boldsymbol{\beta}}_j^\top \boldsymbol{Z}_i(s)}}{\sum_{i=1}^n Y_i(s)e^{\widetilde{\boldsymbol{\beta}}_j^\top \boldsymbol{Z}_i(s)}} - \frac{\sum_{i=1}^n \boldsymbol{Z}_i(s)Y_i(s)e^{\boldsymbol{\beta}^{o\top} \boldsymbol{Z}_i(s)}}{\sum_{i=1}^n Y_i(s)e^{\boldsymbol{\beta}^{o\top} \boldsymbol{Z}_i(s)}}\biggr\|_\infty \nonumber \\
&\leq 2K_Z\sum_{i=1}^n \frac{Y_i(s)|e^{\widetilde{\boldsymbol{\beta}}_j^\top \boldsymbol{Z}_i(s)} - e^{\boldsymbol{\beta}^{o\top} \boldsymbol{Z}_i(s)}|}{\sum_{\ell=1}^n Y_\ell(s)e^{\widetilde{\boldsymbol{\beta}}_j^{\top}\boldsymbol{Z}_\ell(s)}} \nonumber \\
&\leq 2K_Z(e^{K_Z\|\widehat{\boldsymbol{\beta}} - \boldsymbol{\beta}^o\|_1}-1).
\end{align}
Similarly, for any $s \in \mathcal{T}$,
\begin{align}
\label{Eq:WeightDiff}
\sum_{i=1}^n |w_i(s,\widetilde{\boldsymbol{\beta}}_j) - w_i(s,\boldsymbol{\beta}^o)| &= \sum_{i=1}^n \biggl|\frac{Y_i(s)e^{\widetilde{\boldsymbol{\beta}}_j^\top \boldsymbol{Z}_i(s)}}{\sum_{\ell=1}^n Y_\ell(s)e^{\widetilde{\boldsymbol{\beta}}_j^\top \boldsymbol{Z}_\ell(s)}} - \frac{Y_i(s)e^{\boldsymbol{\beta}^{o\top} \boldsymbol{Z}_i(s)}}{\sum_{\ell=1}^n Y_\ell(s)e^{\boldsymbol{\beta}^{o\top}\boldsymbol{Z}_\ell(s)}}\biggr| \nonumber \\
&\leq 2\sum_{i=1}^n \frac{Y_i(s)|e^{\widetilde{\boldsymbol{\beta}}_j^\top \boldsymbol{Z}_i(s)} - e^{\boldsymbol{\beta}^{o\top} \boldsymbol{Z}_i(s)}|}{\sum_{\ell=1}^n Y_\ell(s)e^{\widetilde{\boldsymbol{\beta}}_j^{\top}\boldsymbol{Z}_\ell(s)}} \leq 2\bigl(e^{K_Z\|\widehat{\boldsymbol{\beta}} - \boldsymbol{\beta}^o\|_1} - 1\bigr).
\end{align}
It follows from~\eqref{Eq:Mj},~\eqref{Eq:ZbarDiff2} and~\eqref{Eq:WeightDiff} that 
\begin{equation}
\label{Eq:Mj2}
\|\boldsymbol{M}(\widetilde{\boldsymbol{\beta}}) - \ddot{\ell}(\boldsymbol{\beta}^o)\|_\infty = \max_{j=1,\ldots,p}\|\ddot{\ell}_j(\widetilde{\boldsymbol{\beta}}_j) - \ddot{\ell}_j(\boldsymbol{\beta}^o)\|_\infty \leq 16K_Z^2(e^{K_Z\|\widehat{\boldsymbol{\beta}} - \boldsymbol{\beta}^o\|_1} - 1).
\end{equation}

Moreover,
\begin{align}    
\label{Eq:lddot}
\|\ddot{\ell}(\boldsymbol{\beta}^o) + \boldsymbol{\Sigma}\|_\infty &\leq \biggl\|\frac{1}{n}\sum_{i=1}^n \int_{\mathcal{T}} \sum_{j=1}^n \bigl\{\boldsymbol{Z}_j(s) - \bar{\boldsymbol{Z}}(s,\boldsymbol{\beta}^o)\bigr\}^{\otimes 2} w_j(s,\boldsymbol{\beta}^o) \, dM_i(s)\biggr\|_\infty \nonumber \\
&\hspace{1cm}+ \biggl\|\frac{1}{n}\sum_{i=1}^n \int_0^{t_+} \bigl\{\boldsymbol{Z}_i(s) - \bar{\boldsymbol{Z}}(s,\boldsymbol{\beta}^o)\bigr\}^{\otimes 2} \tilde{w}_i(s,\boldsymbol{\beta}^o) \lambda_0(s) \, ds - \boldsymbol{\Sigma} \biggr\|_\infty.
\end{align}
The first term is the entrywise maximum absolute norm of a random $p \times p$ matrix.  Fixing $j,k \in \{1,\dots,p\}$, it is convenient to write its $(j, k)$th entry as $n^{-1}\sum_{i=1}^n \int_{\mathcal{T}}a(s)\, dM_i(s)$, with
\[
	a(s) = a_{j, k}(s) := \sum_{i=1}^n \bigl\{Z_{ij}(s) - \bar{Z}_j(s,\boldsymbol{\beta}^o)\bigr\}\bigl\{Z_{ik}(s) - \bar{Z}_k(s,\boldsymbol{\beta}^o)\bigr\} w_i(s,\boldsymbol{\beta}^o),
\]
where $Z_{ij}(s)$ and $\bar{Z}_j(s,\boldsymbol{\beta}^o)$ are the $j$th components of $\boldsymbol{Z}_i(s)$ and $\bar{\boldsymbol{Z}}(s,\boldsymbol{\beta}^o)$ respectively.  For $t \in \mathcal{T}$, we also define the right-continuous martingale
\[
  W_t := \frac{1}{n}\sum_{i=1}^n \int_{[0,t]} a(s)\, dM_i(s),
\]
and claim that $(W_t:t \in \mathcal{T})$ is uniformly integrable.  To see this, note that
\begin{align*}
\sup_{t \in \mathcal{T}} \mathbb{E}(W_t^2) &= \frac{1}{n}\sup_{t \in \mathcal{T}} \mathbb{E} \int_0^t a(s)^2 Y(s)e^{\boldsymbol{\beta}^{o\top}\boldsymbol{Z}(s)}\lambda_0(s) \, ds \\
                                           &\leq \frac{16K_Z^4}{n}\sup_{t \in \mathcal{T}} \mathbb{E}\bigl[\mathbb{E}\{-\log \bar{F}_{\tilde{T}|\boldsymbol{\mathcal{Z}}}(T \wedge t)|\boldsymbol{\mathcal{Z}}\}\bigr] \\
  &\leq \frac{16K_Z^4}{n} \mathbb{E}\bigl[\mathbb{E}\bigl\{-\log \bar{F}_{\tilde{T}|\boldsymbol{\mathcal{Z}}}(\tilde{T})|\boldsymbol{\mathcal{Z}}\bigr\}\bigr] = \frac{16K_Z^4}{n},
\end{align*}
which establishes the desired uniform integrability.  Thus, by, e.g., \citet[][p.18]{KaratzasShreve1991} there exists a random variable $W_{t_+}$ such that $\mathbb{E}(|W_{t_+}|) < \infty$ and $W_t \stackrel{\mathrm{a.s.}}{\rightarrow} W_{t_+}$ as $t \rightarrow t_+$.  
Now let $t_0 := 0$, let $t_j := \inf\bigl\{t \in \mathcal{T}:\sum_{i=1}^n \mathbbm{1}_{\{T_i \leq t\}} = j\bigr\}$ be the $j$th observed or censored event for $j = 1, \ldots, n$, and let $t_{n+1} := t_+$.  Then $\{t_j\}$ is a sequence of increasing stopping times.  For $j = 0, \ldots, n+1$, define $X_j := W_{t_j}$, as well as the $\sigma$-algebra $\mathcal{F}_j = \mathcal{F}_{t_j}$ consisting of those events $A$ for which $A \cap \{t_j \leq t\} \in \mathcal{F}_t$ for every $t \in \mathcal{T}$.  Then, writing $d_j := X_j - X_{j-1}$, we have by the optional sampling theorem \citep[e.g.][Theorem~1.3.22]{KaratzasShreve1991} that $\{d_j:j = 1, \ldots, n+1\}$ is a martingale difference sequence with respect to the filtration $\{\mathcal{F}_j:j=0,1,\ldots,n+1\}$.

We seek to control $\mathbb{E}(|d_j|^k | \mathcal{F}_{j-1})$ for $k \in \mathbb{N}$ with $k \geq 2$.  Writing $s_j := \min_{\ell:\ell \in R_{t_{j-1}}}\tilde{T}_\ell$ for $j=1,\ldots,n$ and $s_{n+1} := t_+$, note that $d_{n+1} = 0$ and for $j=1,\ldots,n$,
	\begin{align*}
	|d_j| &= \frac{1}{n}\biggl|\sum_{i=1}^n \int_{(t_{j-1}, t_j]}a(s)\, dM_i(s)\biggr| = \frac{1}{n}\biggl|\sum_{i=1}^n \int_{(t_{j-1}, t_j]}a(s)\, dN_i(s) - \sum_{i=1}^n \int_{(t_{j-1}, t_j]}a(s)\, d\Lambda_i(s,\boldsymbol{\beta}^o)\biggr| \\
	&\leq \frac{4K_Z^2}{n}\biggl(1 + \sum_{i=1}^n \int_{(t_{j-1}, t_j]}\,d\Lambda_i(s,\boldsymbol{\beta}^o)\biggr) \leq \frac{4K_Z^2}{n}\biggl(1 + \sum_{i=1}^n \int_{(t_{j-1}, s_j]}\,d\Lambda_i(s,\boldsymbol{\beta}^o)\biggr),
	\end{align*}
where the final inequality follows because for every $j = 1, \ldots, n$, if $t_j$ is the time of a censored event, then $s_j > t_j$; if $t_j$ is the time of an observed event, then $s_j = t_j$.  Now let $i^* \in \{1,\ldots,n\}$ denote the smallest index in $R_{t_{j-1}}$, so that $i^*$ is $\mathcal{F}_{t_{j-1}}$-measurable.  Then
\begin{align*}
          \sum_{i=1}^n \int_{(t_{j-1}, s_j]}\,d\Lambda_i(s,\boldsymbol{\beta}^o) &= \sum_{i \in R_{t_{j-1}}} \int_{(t_{j-1}, s_j]} Y_i(s)e^{\boldsymbol{\beta}^{o\top}\boldsymbol{Z}_i(s)}\lambda_0(s) \, ds \\
                                                                                 &\leq e^{2\|\boldsymbol{\beta}^o\|_1K_Z} (n-j+1)\int_{(t_{j-1}, s_j]} e^{\boldsymbol{\beta}^{o\top}\boldsymbol{Z}_{i^*}(s)}\lambda_0(s) \, ds.
\end{align*}
But, writing $\boldsymbol{\mathcal{Z}}^{(n)} := \{\boldsymbol{Z}_i(t): i =1,\ldots,n,t \in \mathcal{T}\}$, for any $x > 0$,
\begin{align*}
  \mathbb{P}\biggl(\int_{(t_{j-1}, s_j]} \! \! \! e^{\boldsymbol{\beta}^{o\top}\boldsymbol{Z}_{i^*}(s)}&\lambda_0(s) \, ds > x \biggm| \mathcal{F}_{j-1}, \boldsymbol{\mathcal{Z}}^{(n)}\biggr) = \mathbb{P}\biggl(-\log \bar{F}_{\tilde{T}_{i^*}|\mathcal{F}_{j-1}, \boldsymbol{\mathcal{Z}}^{(n)}}(s_j) > x \biggm| \mathcal{F}_{j-1}, \boldsymbol{\mathcal{Z}}^{(n)}\biggr) \\
&= \mathbb{P}\biggl(\min_{\ell:\ell \in R_{t_{j-1}}} \tilde{T}_\ell > \bar{F}_{\tilde{T}_{i^*}|\mathcal{F}_{j-1}, \boldsymbol{\mathcal{Z}}^{(n)}}^{-1}(e^{-x}) \biggm| \mathcal{F}_{j-1}, \boldsymbol{\mathcal{Z}}^{(n)}\biggr) \\
                                                                                                           &= e^{-x}\prod_{\ell \in R_{t_{j-1}} \setminus \{i^*\}} \exp\biggl\{-\int_{t_{j-1}}^{\bar{F}_{\tilde{T}_{i^*}|\mathcal{F}_{j-1}, \boldsymbol{\mathcal{Z}}^{(n)}}^{-1}(e^{-x})} e^{\boldsymbol{\beta}^{o\top}\boldsymbol{Z}_\ell(s)} \lambda_0(s) \, ds\biggr\} \\
                                                                                              &\leq e^{-x}\prod_{\ell \in R_{t_{j-1}} \setminus \{i^*\}} \exp\biggl\{-e^{-2\|\boldsymbol{\beta}^o\|_1K_Z}\int_{t_{j-1}}^{\bar{F}_{\tilde{T}_{i^*}|\mathcal{F}_{j-1}, \boldsymbol{\mathcal{Z}}^{(n)}}^{-1}(e^{-x})} e^{\boldsymbol{\beta}^{o\top}\boldsymbol{Z}_{i^*}(s)} \lambda_0(s) \, ds\biggr\} \\
  &\leq \exp\Bigl\{-(n-j+1)e^{-2\|\boldsymbol{\beta}^o\|_1K_Z}x\Bigr\}.
\end{align*}
We deduce that
\[
  \sum_{i=1}^n \int_{(t_{j-1}, s_j]}\,d\Lambda_i(s,\boldsymbol{\beta}^o)\biggm| \mathcal{F}_{j-1} \leq_{\mathrm{st}} e^{2\|\boldsymbol{\beta}^o\|_1K_Z} (n-j+1) \mathrm{Exp}\bigl((n-j+1)e^{-2\|\boldsymbol{\beta}^o\|_1K_Z}\bigr),
\]
where $\leq_{\mathrm{st}}$ denotes the usual stochastic ordering.  In particular,
\begin{align*}
	\mathbb{E}(|d_j|^k | \mathcal{F}_{j-1}) &\leq  \Bigl(\frac{4K_Z^2}{n}\Bigr)^k \sum_{l = 0}^k \binom{k}{l} \mathbb{E}\biggl\{\biggl(\sum_{i=1}^n \int_{(t_{j-1}, s_j]}\,d\Lambda_i(s,\boldsymbol{\beta}^o)\biggr)^l \bigg| \mathcal{F}_{j-1}\biggr\} \nonumber \\
             &\leq \Bigl(\frac{4K_Z^2}{n}\Bigr)^k e^{4k\|\boldsymbol{\beta}^o\|_1K_Z} \sum_{l = 0}^k \binom{k}{l} l! = \Bigl(\frac{4K_Z^2}{n}\Bigr)^k e^{4k\|\boldsymbol{\beta}^o\|_1K_Z}k!e.
	\end{align*}
Hence, for all $k \in \mathbb{N}$, we have
\begin{align}\label{eq-dj-2k}
	& \mathbb{E}(|d_j|^{2k} | \mathcal{F}_{j-1}) \leq  \Bigl(\frac{4e^{4\|\boldsymbol{\beta}^o\|_1K_Z+1}K_Z^2}{n}\Bigr)^{2k} (2k)!.
\end{align}
From~\eqref{eq-dj-2k}, and writing $\nu := 8e^{4\|\boldsymbol{\beta}^o\|_1K_Z+1}K_Z^2$, we can apply \citet[][Theorem~2.3]{BoucheronEtal2013} and the fact that $d_1,\ldots,d_{n+1}$ have zero mean to deduce that each $d_j|\mathcal{F}_{j-1}$ is a sub-gamma random variable with parameters $\nu^2/n^2$ and $\nu/n$.  Now let $\mathcal{G} := \sigma(d_1,\ldots,d_{n+1})$.  It follows from Section 2 of \cite{pena1999} that for the sequence $(d_j)$, there exists a tangent sequence $(e_j)$ satisfying
	\[
	d_j | \mathcal{F}_{j-1} \stackrel{d}{=} e_j | \mathcal{F}_{j-1} \stackrel{d}{=} e_j | \mathcal{G}
	\]
	and such that $e_1,\ldots,e_{n+1}$ are conditionally independent given $\mathcal{G}$.
Thus, for $x > 0$,
\begin{align*}
\mathbb{P}\biggl(\sum_{j=1}^{n+1}d_j  \geq x \biggr)	& \leq \inf_{s > 0} e^{-sx}\mathbb{E}\biggl\{\exp\biggl(s\sum_{j=1}^n d_j\biggr)\biggr\} \leq \inf_{s > 0}e^{-sx}\biggl\{\mathbb{E}\exp \biggl(2s\sum_{j=1}^n e_j\biggr)\biggr\}^{1/2} \\
                                                        & = \inf_{s > 0}e^{-sx}\biggl[\mathbb{E}\biggl\{\mathbb{E}\exp \biggl(2s\sum_{j=1}^n e_j\biggm|\mathcal{G}\biggr)\biggr\}\biggr]^{1/2} \\
& \leq \inf_{0 <s < n/\nu} \exp\biggl(-sx + \frac{\nu^2 s^2}{n - 2\nu s}\biggr) \leq \exp \biggl\{-\frac{nx^2}{4(\nu x + \nu^2)}\biggr\},
\end{align*}
where the second inequality follows from Corollary 3.1 in \cite{pena1999}, the third inequality follows from the conditional independence of the sequence $(e_j)$ and the sub-gamma tail behaviour, and the last inequality holds by taking 
\[
s = \frac{nx}{2(\nu x + \nu^2)} < \frac{n}{\nu}.
\]
Therefore, for $x > 0$,
	\begin{align*}
	\mathbb{P}\biggl\{\biggl\|\frac{1}{n}\sum_{i=1}^n \int_{\mathcal{T}} \sum_{j=1}^n \bigl\{\boldsymbol{Z}_j(s) - \bar{\boldsymbol{Z}}(s,\boldsymbol{\beta}^o)\bigr\}^{\otimes 2} &w_j(s,\boldsymbol{\beta}^o) \, dM_i(s)\biggr\|_\infty \geq x\biggr\} \leq 2p^2 \exp  \biggl\{-\frac{nx^2}{4(\nu x + \nu^2)}\biggr\}.  
	\end{align*}
We deduce that 
\begin{equation}
\label{Eq:Words1}
\biggl\|\frac{1}{n}\sum_{i=1}^n \int_{\mathcal{T}} \sum_{j=1}^n \bigl\{\boldsymbol{Z}_j(s) - \bar{\boldsymbol{Z}}(s,\boldsymbol{\beta}^o)\bigr\}^{\otimes 2} w_j(s,\boldsymbol{\beta}^o) \, dM_i(s)\biggr\|_\infty = O_p\biggl(\frac{\log^{1/2}(np)}{n^{1/2}}\biggr).
\end{equation}
For the second term in~\eqref{Eq:lddot}, observe that 
\begin{align}
\label{Eq:CLT}
\biggl\|\frac{1}{n}\sum_{i=1}^n\int_0^{t_+} \bigl\{\boldsymbol{Z}_i(s) &- \bar{\boldsymbol{Z}}(s, \boldsymbol{\beta}^o)\bigr\}^{\otimes 2}\tilde{w}_i(s,\boldsymbol{\beta}^o)\lambda_0(s)\,ds - \boldsymbol{\Sigma}\biggr\|_{\infty} \nonumber \\
&\leq \biggl\|\frac{1}{n}\sum_{i=1}^n\int_0^{t_+} \bigl\{\boldsymbol{Z}_i(s) - \boldsymbol{\mu}(s,\boldsymbol{\beta}^o)\bigr\}^{\otimes 2}\tilde{w}_i(s,\boldsymbol{\beta}^o)\lambda_0(s)\,ds - \boldsymbol{\Sigma}\biggr\|_{\infty} \nonumber \\
&\hspace{2cm}+ \frac{1}{n}\sum_{i=1}^n\int_0^{t_+} \bigl\|\bar{\boldsymbol{Z}}(s, \boldsymbol{\beta}^o)- \boldsymbol{\mu}(s,\boldsymbol{\beta}^o)\bigr\|_\infty^2\tilde{w}_i(s,\boldsymbol{\beta}^o)\lambda_0(s)\,ds.
\end{align}
For the first term in~\eqref{Eq:CLT}, we note that it is the maximum absolute value of a random vector, each of whose components is a sample average of independent and identically distributed random variables that are bounded in absolute value by $8K_Z^2$ and have expectation zero.  Thus
\begin{align}
\label{Eq:Words2}
\biggl\|\frac{1}{n}\sum_{i=1}^n\int_0^{t_+} \bigl\{\boldsymbol{Z}_i(s) - \boldsymbol{\mu}(s,\boldsymbol{\beta}^o)\bigr\}^{\otimes 2}\tilde{w}_i(s,\boldsymbol{\beta}^o)\lambda_0(s)\,ds - \boldsymbol{\Sigma}\biggr\|_{\infty} = O_p\biggl(\frac{\log^{1/2} (np)}{n^{1/2}}\biggr).
\end{align}
For the second term in~\eqref{Eq:CLT}, by~\eqref{Eq:InftyBound},~\eqref{Eq:S1Bound},~\eqref{Eq:InftyBound2} and Markov's inequality, we have that for any $\delta > 0$,
\begin{equation}
\label{Eq:delta}
\frac{1}{n}\sum_{i=1}^n\int_0^{t_+} \bigl\|\bar{\boldsymbol{Z}}(s, \boldsymbol{\beta}^o)- \boldsymbol{\mu}(s,\boldsymbol{\beta}^o)\bigr\|_\infty^2\tilde{w}_i(s,\boldsymbol{\beta}^o)\lambda_0(s)\,ds = o_p(n^{-(1/2-\delta)}).
\end{equation}
We deduce from~\eqref{Eq:lddot},~\eqref{Eq:Words1},~\eqref{Eq:CLT},~\eqref{Eq:Words2} and~\eqref{Eq:delta} that for every $\delta > 0$,
\begin{equation}
\label{Eq:lddotConc}
\|\ddot{\ell}(\boldsymbol{\beta}^o) + \boldsymbol{\Sigma}\|_\infty = o_p(n^{-(1/2 - \delta)}).
\end{equation}

Combining~\eqref{Eq:Start} with 
\eqref{eq-prop2-im} in the proof of Proposition~\ref{prop-tb1},~\eqref{Eq:Mj2},~\eqref{Eq:lddotConc}, Lemma~\ref{Lem:lasso}(ii), \textbf{(A4)(b)} and \textbf{(A4)(c)}, we have
\[
\bigl|\boldsymbol{c}^{\top}\bigl(\widehat{\boldsymbol{\Theta}}\boldsymbol{M}(\widetilde{\boldsymbol{\beta}}) + \boldsymbol{I}\bigr)(\widehat{\boldsymbol{\beta}} - \boldsymbol{\beta}^o)\bigr| \leq \|\widehat{\boldsymbol{\Theta}}\boldsymbol{M}(\widetilde{\boldsymbol{\beta}}) + \boldsymbol{I}\|_{\infty}\|\widehat{\boldsymbol{\beta}} - \boldsymbol{\beta}^o\|_1 = o_p(n^{-1/2}),
\]
as required.
\end{proof}

\subsubsection{Completion of the proof}

We now wrap up all the results in the previous three subsections.
\begin{proof}[Proof of Theorem~\ref{thm-main}]
  From~\eqref{Eq:1}, Proposition~\ref{lem-leading}, Proposition~\ref{prop-tb1} and~\ref{prop-tb2}, we deduce from Slutsky's theorem that under the stated assumptions, the first claim follows.  To prove the second claim, note that 

	\[
	\bigl|\boldsymbol{c}^{\top}\widehat{\boldsymbol{\Theta}}\boldsymbol{c} - \boldsymbol{c}^{\top}\boldsymbol{\Sigma}^{-1}\boldsymbol{c}\bigr| \leq \bigl\|\widehat{\boldsymbol{\Theta}} - \boldsymbol{\Sigma}^{-1}\bigr\|_{\infty} = o_p(1),
	\]
where the final claim follows from~\eqref{eq-prop2-im}, Lemma~\ref{lem-aclime-lem1} and \textbf{(A4)(c)}.  Another application of Slutsky's theorem therefore yields the second claim.  
\end{proof}

\section{Numerical experiments}\label{sec:numerical}

In this section, we investigate the numerical performance of our proposed method.  We begin by discussing various practical implementation issues in~\Cref{sec-pracissue}; in Sections~\ref{sec-simulations} and~\ref{sec-realdata}, we present analyses of simulated data and real data, respectively.

\subsection{Practical issues}\label{sec-pracissue}




\subsubsection{Software}

Recall that the debiased estimator $\widehat{\boldsymbol{b}}$ is obtained from a Lasso estimator $\widehat{\boldsymbol{\beta}}$ of the vector of true regression coefficients $\boldsymbol{\beta}^o = (\beta_1^o,\ldots,\beta_p^o)^\top$, as well as a CLIME-type estimator $\widehat{\boldsymbol{\Theta}}$ of $\boldsymbol{\Sigma}^{-1}$, the population version of the inverse of the negative Hessian matrix.  
We use the R \citep{R} package \textsc{glmnet} \citep{FriedmanEtal2010, SimonEtal2011} to compute $\widehat{\boldsymbol{\beta}}$; and adapt the \textsc{clime} \citep{CaiEtal2012} and \textsc{flare} \citep{LiEtal2014} packages to obtain $\widehat{\boldsymbol{\Theta}}$.  The \textsc{clime} package is more accurate, but is slow to compute for high-dimensional data; the \textsc{flare} algorithm computes only an approximate solution, but is faster.  For simplicity, we will refer to the modified \textsc{clime} and \textsc{flare} algorithms as the \textsc{clime} and \textsc{flare} packages, respectively.  In fact we also conducted analysis based on unmodified \textsc{clime} and \textsc{flare} (with \texttt{sym = `or'}) packages, and the differences were negligible.  


\subsubsection{Tuning parameters}

Our theoretical results provide conditions on the tuning parameters $\lambda$ and $\lambda_n$ under which our confidence intervals are asymptotically valid; however, in practice, the unknown population quantities and the unspecified constants mean that these conditions do not provide a practical algorithm for choosing these tuning parameters.  Therefore, to choose $\lambda$, we use the default 10-fold cross-validation algorithm implemented in the \textsc{glmnet} package, with a grid of 100 different tuning parameters, equally spaced on the log scale.
When using the \textsc{clime} and \textsc{flare} packages to compute $\widehat{\boldsymbol{\Theta}}$, the default 10-fold cross-validation algorithms were used to compute $\lambda_n$, with $\mathrm{tr}\bigl(\mathrm{diag}\bigl((\widehat{\boldsymbol{\Sigma}}\widehat{\boldsymbol{\Theta}} - \boldsymbol{I})^2\bigr)\bigr)$ as the cross-validation criterion.  

\subsubsection{Covariates}\label{sec-Z}

Assumption~\textbf{(A1)}(i) asks that the covariate process $\boldsymbol{\mathcal{Z}}$ should be bounded.  However, in our numerical results, we generate the covariate processes from a multivariate Gaussian distribution, due to the convenience of generating different correlation structures.  We also focus for simplicity on time-independent covariates.

An important observation is that even if $\boldsymbol{Z} = (Z_1,\ldots,Z_p)^\top$ has identity covariance matrix, this does not necessarily mean that $\boldsymbol{\Sigma} = (\Sigma_{ij})$ is the identity matrix.  We can illustrate this in the case where
$\boldsymbol{Z} \sim \mathcal{N}_p(\boldsymbol{0},\boldsymbol{\Sigma}^Z)$, as follows: suppose that $(\boldsymbol{\Sigma}^Z)_{ij} = 0$ whenever $\beta^o_i \neq 0$ and $\beta^o_j = 0$.  Then
\begin{itemize}
\item for any $i,j$ with $\beta^o_i \neq 0$ and $\beta^o_j = 0$, we have $\Sigma_{ij} = 0$;
\item for any $i,j$ with $\beta^o_i = 0$ and $\beta^o_j = 0$, we have  
	\[
	\Sigma_{ij} = \mathbb{E}(Z_iZ_j)\mathbb{E}\int_0^{t_+}Y(s)\exp\biggl(\sum_{l: \beta^o_l \neq 0}\beta^o_lZ_l\biggr)\lambda_0(s)\, ds;
	\]
\item for any $i,j$ with $\beta^o_i \neq 0$ and $\beta^o_j = 0$, we have
	\begin{align*}
          \Sigma_{ij} = \mathbb{E}\int_0^{t_+} c_i(s)c_j(s)Y(s)\exp\biggl(\sum_{l: \beta^o_l \neq 0}\beta^o_lZ_l\biggr)\lambda_0(s)\, ds,
        \end{align*}
        where
        \[
          c_i(s) := Z_i - \frac{\mathbb{E}\Bigl\{Z_i Y(s)\exp\Bigl(\sum_{l: \beta^o_l \neq 0}\beta^o_lZ_l\Bigr)\Bigr\}}{\mathbb{E}\Bigl\{Y(s)\exp\Bigl(\sum_{l: \beta^o_l \neq 0}\beta^o_lZ_l\Bigr)\Bigr\}}.
        \]
\end{itemize}
In order to satisfy the sparse precision matrix conditions, we consider the following two choices of $\boldsymbol{\Sigma}^Z$ in our simulations in Section~\ref{sec-simulations}.
\begin{enumerate}
\item [a.] $\boldsymbol{\Sigma}^Z_a = \boldsymbol{I}$;
\item [b.] $\boldsymbol{\Sigma}^Z_b = (\Sigma^Z_b)_{ij}$ with  
	\[
	(\Sigma_b^Z)_{ij} = \begin{cases}
	1, & \mbox{if } i = j,\\
 	0.5, & \mbox{if } i \neq j, \, \beta^o_i \neq 0, \beta^o_j \neq 0, \\
 	0, & \mbox{if } i \neq j, \beta^o_i\beta^o_j = 0, |\beta^o_i| + |\beta^o_j| > 0,\\
 	0.5^{|i-j|}, & \mbox{if } i \neq j, \beta^o_i = 0, \beta^o_j = 0.
	\end{cases}
	\]
\end{enumerate}

\subsubsection{A simple preliminary example}

To illustrate several of the features that arise in more complicated settings, we consider the following two scenarios: let $n = 1000$; $p=10$; $\boldsymbol{Z} \sim \mathcal{N}_p(\boldsymbol{0}, \boldsymbol{I})$; $\beta_1^o = \cdots = \beta_d^o = 1$, and $\beta_{d+1}^o = \cdots = \beta_p^o = 0$ for $d = 1, 3$; $\lambda_0(t) = 1$, for all $t > 0$; $U_i = 3$ when $d_o = 1$ and $U_i = 5$ when $d_o = 3$.  In these settings, the average censoring rate is around 15\%.  In the top-left blocks of Tables~\ref{tab-1-asy} and~\ref{tab-2-asy}, we report the average initial estimator error $\hat{\beta}_j - \beta^o_j$ for each index $j=1,\ldots,p$, the average debiased estimator error $\hat{b}_j - \beta^o_j$, the average of empirical coverage (EC) of the 95\% confidence intervals, their average widths, and the average $p$-values, based on 100 repetitions.  Standard errors for all quantities are given in brackets.

Here, the results are quite encouraging: the biases of the estimates $\hat{\beta}_j$ of the signal variables are substantially corrected by the debiased estimator $\hat{b}_j$, the coverage probabilities are satisfactory (certainly in the $d_o = 1$ case) and the $p$-values for the noise variables appear to be approximately uniformly distributed (notice that, under uniformity, the standard errors should be close to $1/(100 \times 12)^{1/2} \approx 0.029$).  Of course, this is a setting in which the usual inference for the maximum partial likelihood estimate (MPLE) is also valid, as illustrated in the bottom-right blocks of~Tables~\ref{tab-1-asy} and~\ref{tab-2-asy} (for ease of exposition, the MPLE estimators are collected in the $\hat{b}_j - \beta_j^o$ columns).  The MPLE was computed using the package \textsc{survival} \citep{Therneau2015}.

Closer inspection, however, reveals that the situation is not perhaps ideal as it seems at first sight.  First, while the bias correction works very well for the noise variables, it slightly under-corrects for the noise variables.  Second, the widths of the confidence intervals are slightly smaller than those for the MPLE, which is an efficient estimator.  These issues both arise from our choice of precision matrix estimator $\widehat{\boldsymbol{\Theta}}$, which aims to provide a good approximation to $\boldsymbol{\Sigma}^{-1}$ in different matrix norms.  To attempt to address this, we therefore consider widening the intervals by replacing the diagonal entries of $\widehat{\boldsymbol{\Theta}}$ in~\eqref{eq-hatci} with the diagonal entries of $\widetilde{\boldsymbol{\Theta}}$, where $\widetilde{\boldsymbol{\Theta}} = (\widetilde{\Theta}_{ij}) \in \mathbb{R}^{p \times p}$ is given by
	\begin{equation}\label{eq-ThetaTilde}
	\widetilde{\Theta}_{ij} = \begin{cases}
 	\widehat{\Theta}_{ij} & \mbox{if } i\neq j; \\ 	
 	\max\{1/\widehat{\mathcal{V}}(\widehat{\boldsymbol{\beta}})_{jj}, \widehat{\Theta}_{jj}\} & \mbox{if } i = j.
 	\end{cases}
	\end{equation}
The rationale behind our definition of $\widetilde{\boldsymbol{\Theta}}$ is that in an extreme case, when $\widehat{\mathcal{V}}(\widehat{\boldsymbol{\beta}})$ is a diagonal matrix, $\widehat{\boldsymbol{\Theta}}$ is still a biased estimator of $\boldsymbol{\Sigma}^{-1}$.  Since our precision matrix estimators are also potentially sensitive to tuning parameter choice, and the default choice tends to over-penalise, we further consider alternative options to the 10-fold cross-validation choice $\lambda_{\mathrm{CV}}$ in the other blocks of Tables~\ref{tab-1-asy} and~\ref{tab-2-asy}:
\begin{itemize}
\item [(1)] Top-right: $\widehat{\boldsymbol{\Theta}}, 0.1\lambda_{\mathrm{CV}}$ -- confidence interval constructed based on~\eqref{eq-hatci} with $0.1\lambda_{\mathrm{CV}}$ used in $\widehat{\boldsymbol{\Theta}}$, which is provided by the \textsc{clime} package;
  	\item [(2)] Middle-left: $\widetilde{\boldsymbol{\Theta}}$ -- confidence interval replaces $\widehat{\boldsymbol{\Theta}}$ in~\eqref{eq-hatci} with $\widetilde{\boldsymbol{\Theta}}$, computed using~\eqref{eq-ThetaTilde} with $\lambda_{\mathrm{CV}}$ in the \textsc{clime} package;
	\item [(3)] Middle-right: $\widehat{\boldsymbol{\Theta}}$, \textsc{flare} -- confidence interval constructed based on~\eqref{eq-hatci}, and $\widehat{\boldsymbol{\Theta}}$ is computed using the \textsc{flare} package;
	\item [(4)] Bottom-left: Merge -- confidence interval constructed based on~\eqref{eq-hatci}, the tuning parameter for the sparse precision matrix is provided by the \textsc{flare} package using cross-validation, and $\widehat{\boldsymbol{\Theta}}$ is optimised by the \textsc{clime} package using the previously mentioned tuning parameter.
	\end{itemize}

\begin{landscape}
\begin{table}[htbp]
\centering
{\footnotesize
\begin{tabular}{cc | cccc| ccccc}
\hline
$\beta^o_j$ & $\hat{\beta}_j - \beta^o_j$  & $\hat{b}_j-\beta^o_j$  & EC & Width & $p$-value  & $\hat{b}_j - \beta^o_j$ & EC & Width & $p$-value \\
\hline
 & & \multicolumn{4}{c|}{$\widehat{\boldsymbol{\Theta}}$, $\lambda_{\mathrm{CV}}$} & \multicolumn{4}{c}{$\widehat{\boldsymbol{\Theta}}$, $0.1\lambda_{\mathrm{CV}}$}\\	
\hline 
1 & -.051(.005) & -.003(.005) & .92(.03) & .158(.000) & .000(.000) & -.001(.005)  & .93(.03) & .160(.000) & .000(.000)\\
0 & .000(.002) &  .001(.003) & .96(.02) & .124(.000) & .540(.028) & .001(.003)  & .95(.02) & .125(.000) & .531(.028)\\
0 & .000(.001) & -.000(.003)  & .97(.02) & .123(.000) & .540(.029) & .000(.003)  & .97(.02) & .125(.000) & .534(.029)\\
0 & .001(.002) & .002(.003)  & .92(.03) & .124(.000) & .511(.030) & .002(.003)  & .91(.03) & .125(.000) & .503(.030)\\
0 & -.000(.002) & -.001(.003)  & .94(.02) & .124(.000) & .532(.030) & -.001(.003)  & .95(.02) & .125(.000) & .526(.030)\\
0 & .000(.001) & -.001(.003)  & .99(.01) & .123(.000) & .522(.028) & -.002(.003)  & .99(.01) & .125(.000) & .516(.027)\\
0 &  -.001(.001) & -.004(.003)  & .96(.02) & .124(.000) & .527(.030) & -.005(.003)  & .96(.02) & .125(.000) & .522(.030)\\
0 &  -.001(.002) & -.003(.003)  & .93(.03) & .123(.000) & .516(.030) & -.003(.003)  & .93(.03) & .125(.000) & .510(.030)\\
0 & -.001(.002)  & -.000(.003) & .90(.03) & .123(.000) & .490(.030) & .000(.004)  & .91(.03) & .125(.000) & .482(.029)\\
0& .001(.002) & -.001(.003) & .94(.02) & .123(.000) & .473(.029) & -.001(.004)  & .94(.02) & .125(.000) & .467(.029)\\
\hline
& &  \multicolumn{4}{c|}{$\widetilde{\boldsymbol{\Theta}}$} &   \multicolumn{4}{c}{$\widehat{\boldsymbol{\Theta}}$, \textsc{flare}}\\	
\hline
1 & -.051(.005) & -.003(.005)  & .95(.02) & .171(.000) & .000(.000) & -.007(.005)  & .90(.03) & .153(.000) & .000(.000)\\
0 & .000(.002) & .001(.003) & .97(.02) & .135(.000) & .568(.027) & .001(.003)  & .96(.02) & .121(.000) & .546(.028)\\
0 & .000(.001) & -.000(.003)  & .98(.01) & .134(.000) & .569(.028) & .000(.003)  & .97(.02) & .121(.000) & .545(.029)\\
0 & .001(.002) & .002(.003)  & .94(.02) & .135(.000) & .536(.030) & .002(.003)  & .92(.03) & .121(.000) & .515(.030)\\
0 & -.000(.002) & -.001(.003) & .98(.01) & .134(.000) & .555(.029) & -.001(.003)  & .95(.02) & .121(.000) & .536(.030)\\
0 & .000(.001) & -.001(.003) & .99(.01) & .134(.000) & .553(.026) & -.001(.003)  & .99(.01) & .121(.000) & .529(.028)\\
0 &  -.001(.001) & -.004(.003)  & .99(.01) & .134(.000) & .555(.029) & -.004(.003)  & .98(.01) & .121(.000) & .532(.030)\\
0 &  -.001(.002) & -.003(.003) & .98(.01) & .134(.000) & .548(.029) & -.003(.003)  & .92(.03) & .121(.000) & .525(.030)\\
0 & -.001(.002) & -.000(.003)  & .94(.02) & .134(.000) & .519(.029) & .000(.003)  & .90(.03) & .121(.000) & .497(.030)\\
0& .001(.002) & -.001(.003)& .95(.02) & .134(.000) & .503(.029) & -.001(.003)  & .94(.02) & .121(.000) & .479(.029)\\
\hline
 & & \multicolumn{4}{c|}{Merge} & \multicolumn{4}{c}{MPLE}  \\	
\hline
1 & -.051(.005) & -.004(.005) & .91(.03) & .157(.000) & .000(.000)& .005(.005)  & .95(.22) & .172(.005) & .000(.000)\\
0 & .000(.002) & .001(.003)  & .96(.02) & .123(.000) & .544(.028)& .001(.003) & .95(.22) &.136(.005) & .499(.285)\\
0 & .000(.001) & .000(.003)  & .97(.02) & .123(.000) & .546(.029)& -.001(.003) & .98(.14) &.135(.004) & .513(.282)\\
0 & .001(.002) & .002(.003)  & .91(.03) & .123(.000) & .514(.030) & .001(.004) & .92(.27) &.136(.005) & .489(.301)\\
0 & -.000(.002) & -.001(.003)  & .95(.02) & .123(.000) & .536(.030) & .001(.003) & .94(.24) &.136(.004) & .500(.296)\\
0 & .000(.001) & -.001(.003)  & .99(.01) & .123(.000) & .528(.027)& -.001(.003) & .98(.14) &.136(.004) & .500(.279)\\
0 &  -.001(.001) & -.004(.003)  & .96(.02) & .123(.000) & .528(.029)& -.005(.003) & .95(.22) &.136(.004) & .503(.303)\\
0 &  -.001(.002) & -.003(.003)  & .93(.02) & .123(.000) & .523(.029) & -.004(.003) & .95(.22) &.135(.004) & .512(.307)\\
0 & -.001(.002)  & .000(.003) & .92(.03) & .123(.000) & .493(.030)& .000(.004) & .91(.29) &.136(.004) & .477(.297)\\
0& .001(.002) & -.001(.003)  & .94(.02) & .123(.000) & .476(.029)& .001(.004) & .96(.20) &.135(.004) & .446(.290)\\
\hline
\end{tabular}}
\caption{Simple preliminary example, $d_o = 1$.}\label{tab-1-asy}
\end{table}
\end{landscape}

\begin{landscape}
\begin{table}[htbp]
\centering
{\footnotesize
\begin{tabular}{cc | cccc| cccc}
\hline
$\beta^o_j$ & $\hat{\beta}_j - \beta^o_j$ & $\hat{b}_j - \beta^o_j$ & EC & Wid & pvals  & $\hat{b}_j$  & EC & Width & $p$-value \\
\hline
& & \multicolumn{4}{c|}{$\widehat{\boldsymbol{\Theta}}$, $\lambda_{\mathrm{CV}}$} & \multicolumn{4}{c}{$\widehat{\boldsymbol{\Theta}}$, $0.1\lambda_{\mathrm{CV}}$}\\	
\hline 
1& -.038(.005)  & -.003(.004) & .91(.03) & .153(.000) & .000(.000) & .001(.004) & .92(.03) & .157(.000) & .000(.000)\\
1& -.038(.004)  & -.003(.004)  & .95(.02) & .154(.000) & .000(.000) & .001(.004)  & .97(.02) & .158(.000) & .000(.000)\\
1 & -.039(.004) & -.003(.004) & .93(.03) & .153(.000) & .000(.000) & .000(.004)  & .93(.03) & .157(.000) & .000(.000)\\
0& .000(.003)  & .000(.003) & .91(.03) & .123(.000) & .507(.032) & .000(.004)  & .90(.03) & .125(.000) & .502(.032)\\
0& .000(.003) & .000(.003)  & .93(.03) & .123(.000) & .497(.031) & .000(.003)  & .93(.03) & .125(.000) & .492(.031)\\
0 & -.001(.002)  & -.002(.003)& .93(.03) & .123(.000) & .503(.031) & -.002(.003) & .93(.03) & .125(.000) & .498(.031)\\
0& -.004(.002)  & -.004(.003) & .94(.02) & .123(.000) & .521(.030) & -.004(.003)& .95(.02) & .125(.000) & .516(.030)\\
0& -.003(.002) & -.005(.003)  & .92(.03) & .123(.000) & .483(.030) & -.005(.003) & .91(.03) & .125(.000) & .479(.030)\\
0& -.002(.003) & -.003(.003)  & .92(.03) & .123(.000) & .485(.030) & -.002(.003) & .92(.03) & .125(.000) & .478(.030)\\
0 & .003(.003) & .002(.004) & .92(.03) & .123(.000) & .451(.029) & .002(.004) & .92(.03) & .125(.000) & .447(.029)\\
\hline
 & &  \multicolumn{4}{c|}{$\widetilde{\boldsymbol{\Theta}}$} &   \multicolumn{4}{c}{$\widehat{\boldsymbol{\Theta}}$, \textsc{flare}}\\	
\hline
1& -.038(.005)  & -.003(.004) & .91(.03) & .154(.000) & .000(.000) & -.006(.004) & .91(.03) & .149(.000) & .000(.000)\\
1& -.038(.004) & -.003(.004) & .97(.02) & .155(.000) & .000(.000) & -.006(.004)  & .95(.02) & .149(.000) & .000(.000)\\
1 & -.039(.004) & -.003(.004) & .92(.03) & .154(.000) & .000(.000) & -.007(.004) & .89(.03) & .148(.000) & .000(.000)\\
0& .000(.003)  & .000(.003) & .93(.03) & .136(.000) & .534(.031) & .000(.003) & .91(.03) & .121(.000) & .509(.032)\\
0& .000(.003)  & .000(.003)  & .98(.01) & .135(.000) & .524(.030) & .000(.003) & .94(.02) & .121(.000) & .499(.032)\\
0 & -.001(.002) & -.002(.003) & .97(.02) & .135(.000) & .534(.030) & -.002(.003) & .93(.03) & .121(.000) & .504(.031)\\
0 & -.001(.002)  & -.004(.003) & .95(.02) & .135(.000) & .552(.029) & -.004(.003) & .94(.02) & .121(.000) & .522(.030)\\
0& -.003(.002) & -.005(.003) & .97(.02) & .135(.000) & .515(.029) & -.005(.003) & .92(.03) & .121(.000) & .485(.031)\\
0& -.002(.003)  & -.003(.003)& .94(.02) & .135(.000) & .515(.029) & -.003(.003) & .92(.03) & .121(.000) & .486(.030)\\
0 & .003(.003)  & .002(.004)& .95(.02) & .135(.000) & .482(.028) & .002(.004) & .91(.03) & .121(.000) & .451(.029)\\
\hline
& & \multicolumn{4}{c|}{Merge} & \multicolumn{4}{c}{MPLE}  \\	
\hline
1& -.038(.005)  & -.003(.004)  & .91(.03) & .154(.000) & .000(.000)& .008(.045)  & .94(.24) & .171(.005) & .000(.000)\\
1& -.038(.004) & -.003(.004)  & .95(.02) & .154(.000) & .000(.000) & .008(.039) & .99(.10) & .172(.005) & .000(.000)\\
1 & -.039(.004) & -.004(.004)  & .93(.03) & .153(.000) & .000(.000) & .003(.041) & .96(.20) & .171(.005) & .000(.000)\\
0& .000(.003)& .000(.003) & .92(.03) & .124(.000) & .510(.032)& .000(.039) & .90(.30) & .137(.005) & .487(.308)\\
0& .000(.003) & .000(.003) & .94(.02) & .123(.000) &.497(.031) & .001(.036) & .95(.22) & .136(.004) & .484(.301)\\
0 & -.001(.002)   & -.002(.003)  & .93(.03) & .123(.000) & .504(.031) & -.002(.034) & .97(.17) & .137(.004) & .497(.299)\\
0 & -.001(.002) & -.004(.003)  & .95(.02) & .123(.000) & .523(.030) & -.005(.034) & .95(.22) & .137(.004) & .508(.294)\\
0& -.003(.002) & -.005(.003)  & .92(.03) & .123(.000) & .485(.030) & -.006(.035) & .96(.20) & .136(.004) & .491(.295)\\
0& -.002(.003)  & -.003(.003)  & .92(.03) & .123(.000) & .486(.029) & .000(.038) &.93(.26) & .137(.004) & .464(.285)\\
0 & .003(.003) & .002(.004)  & .92(.03) & .123(.000) & .452(.029) & .001(.039)  & .94(.24) & .136(.004) & .443(.284)\\
\hline
\end{tabular}}
\caption{Simple preliminary example, $d_o = 3$.}\label{tab-2-asy}
\end{table}
\end{landscape}


Comparing the columns of $\hat{\beta}_j - \beta^o_j$ and $\hat{b}_j - \beta^o_j$, we can see that our proposed methods indeed correct the bias due to the shrinkage introduced by the Lasso estimators, but the biases for the signal variables are not fully corrected, and in terms of the signs of the errors, they all tend to be under-corrected, except the $\widehat{\boldsymbol{\Theta}}, 0.1\lambda_{\mathrm{CV}}$ blocks.  The differences between the $\widehat{\boldsymbol{\Theta}}, \lambda_{\mathrm{CV}}$ and $\widehat{\boldsymbol{\Theta}}, 0.1\lambda_{\mathrm{CV}}$ blocks show that the 10-fold cross-validation chosen tuning parameters still over-penalise the sparse precision matrix estimation and lead to under-correction of $\widehat{\boldsymbol{b}}$.  From the EC and Width columns in the $\widehat{\boldsymbol{\Theta}}, \lambda_{\mathrm{CV}}$ and $\widetilde{\boldsymbol{\Theta}}$ blocks, we can see that in some cases, using $\widetilde{\boldsymbol{\Theta}}$ indeed helps in terms of improving the coverages (naturally, the confidence intervals are a little wider).  
We can also see that 
the \textsc{flare} package does not produce identical solutions to the \textsc{clime} package even in this relatively simple context.  It is worth noting that the $\widehat{\boldsymbol{\Theta}}$, \textsc{FLARE} and Merge blocks have the same initial estimators, the same tuning parameter grids for $\widehat{\boldsymbol{\Theta}}$ and the same cross-validation algorithms.  Further investigation in the case $d_o=1$ reveals that the \textsc{flare} package tends to choose slightly larger tuning parameters, which explains the better centering and coverage of the \textsc{clime} confidence intervals; see Table~\ref{tab-lambda}.


\begin{table}[htbp]
\centering
\begin{tabular}{c c c}
\hline
Packages & Mean & Median\\
\hline
\textsc{clime} & 0.022 & 0.015\\
\textsc{flare} & 0.026 & 0.025\\
\hline
\end{tabular}
\caption{Selected tuning parameter comparisons.}\label{tab-lambda}
\end{table}

\subsection{Further simulated examples}\label{sec-simulations}

In order to provide a deeper understanding of our proposed method, we consider the following 16 simulation settings described below, where CT is the censoring time and CR is censoring rate.

\begin{enumerate}
	\item [(1)] $n = 1000$, $p = 10$, $\beta^o_i = 1$, $i = 1, 2, 3$, $\beta^o_i = 0$, $i = 4, \ldots, 10$, $\boldsymbol{Z}\sim \mathcal{N}(0, \boldsymbol{\Sigma}^Z_a)$, $\mathrm{CT} = 5$, $\mathrm{CR}\approx 15\%$;
	\item [(2)] $n = 1000$, $p = 10$, $\beta^o_i = 1$, $i = 1, 2, 3$, $\beta^o_i = 0$, $i = 4, \ldots, 10$, $\boldsymbol{Z}\sim \mathcal{N}(0, \boldsymbol{\Sigma}^Z_a)$, $\mathrm{CT} = 2$, $\mathrm{CR}\approx 30\%$; 
	\item [(3)] $n = 1000$, $p = 10$, $(\beta^o_1, \beta^o_2, \beta^o_3) = (1.2, 1, 0.8)$, $\beta^o_i = 0$, $i = 4, \ldots, 10$, $\boldsymbol{Z}\sim \mathcal{N}(0, \boldsymbol{\Sigma}^Z_a)$, $\mathrm{CT} = 5$, $\mathrm{CR}\approx 15\%$; 
	\item [(4)] $n = 1000$, $p = 10$, $(\beta^o_1, \beta^o_2, \beta^o_3) = (1.2, 1, 0.8)$, $\beta^o_i = 0$, $i = 4, \ldots, 10$, $\boldsymbol{Z}\sim \mathcal{N}(0, \boldsymbol{\Sigma}^Z_a)$, $\mathrm{CT} = 2$, $\mathrm{CR}\approx 30\%$;  
	\item [(5-8)] As for (1)-(4), but with $\boldsymbol{Z} \sim \mathcal{N}(0, \boldsymbol{\Sigma}^Z_b)$, $\mathrm{CT} = 10, 2.5, 10, 2.5$;
	\item [(9-10)] As for (1)-(2), but with $p = 300$, $\beta^o_i = 1$, $i = 1, \ldots, 6$, $\beta^o_i = 0$, $i = 7, \ldots, 300$, $\mathrm{CT} = 9, 2.5$;
	\item [(11-12)] As for (3)-(4), but with $p = 300$, $(\beta^o_1, \ldots, \beta^o_6) = (0.5, 0.7, 0.9, 1.1, 1.3, 1.5)$, $\beta^o_i = 0$, $i = 7, \ldots, 300$, $\mathrm{CT} = 10, 3$;
	\item [(13-16)] As for (9)-(12), but with $\boldsymbol{Z} \sim \mathcal{N}(0, \boldsymbol{\Sigma}^Z_b)$, $\mathrm{CT} = 100, 7, 100, 7$.
\end{enumerate}

In \Cref{tab-3-asy}, we report averaged results for signal and noise variables separately, with $\widehat{\boldsymbol{\Theta}}$ and $\widetilde{\boldsymbol{\Theta}}$ chosen by 10-fold cross-validation.  
The simulations were run on a cluster, each node of which is a Intel(R) Xeon(R) CPU E5-2670 0@2.60GHz machine, with 16 CPUs.  To conduct one repetition of a ($n, p$) = (1000, 300) setting, it took on average 32 minutes.  This explains why we limit our simulations to $p=300$ even though our theory can handle $p \gg n$ settings.

\begin{landscape}
\begin{table}[htbp]
\centering
{\footnotesize
\begin{tabular}{c | cc |ccc| ccc}
\hline
& & & \multicolumn{3}{c|}{$\widehat{\boldsymbol{\Theta}}$} & \multicolumn{3}{c}{$\widetilde{\boldsymbol{\Theta}}$}\\	
  & $\hat{b}_j-\beta^o_j$ & $\hat{\beta}_j-\beta^o_j$ & EC & Width & $p$-values  & EC & Width & $p$-values \\
\hline
(1) S & -.003(.001) & -.038(.001) & .933(.008) & .153(.000) & .000(.000) & .933(.008) & .154(.000) & .000(.000)\\
(1) N & -.001(.001) & -.001(.001) & .929(.008) & .123(.000) & .491(.010) & .956(.006) & .135(.000) & .521(.009)\\
(2) S & -.009(.001) & -.039(.001) & .907(.009) & .150(.000) & .000(.000) & .940(.008) & .165(.000) & .000(.000)\\
(2) N & -.002(.001) & -.001(.001) & .921(.008) & .123(.000) & .503(.010) & .957(.006) & .147(.000) & .556(.009)\\
(3) S & -.003(.001) & -.038(.001) & .940(.007) & .154(.000) & .000(.000) & .940(.007) & .155(.000) & .000(.000)\\
(3) N & -.002(.001) & -.001(.001) & .933(.008) & .123(.000) & .497(.010) & .951(.007) & .135(.000) & .527(.009)\\
(4) S & -.009(.001) & -.039(.001) & .883(.010) & .150(.000) & .000(.000) & .913(.009) & .166(.000) & .000(.000)\\
(4) N & -.002(.001) & -.001(.001) & .914(.009) & .123(.000) & .510(.010) & .957(.006) & .147(.000) & .565(.010)\\
(5) S & -.004(.002) & -.036(.002) & .937(.008) & .177(.000) & .000(.000) & .953(.006) & .194(.000) & .000(.000)\\
(5) N & .000(.001) & .000(.001) & .933(.008) & .152(.000) & .496(.009) & .937(.008) & .152(.000) & .496(.009)\\
(6) S & -.008(.002) & -.035(.002) & .887(.010) & .174(.000) & .000(.000) & .950(.007) & .211(.000) & .000(.000)\\
(6) N & .000(.001) & .000(.001) & .913(.009) & .151(.000) & .495(.01) & .921(.008) & .154(.000) & .508(.010)\\
(7) S & -.003(.002) & -.036(.002) & .930(.008) & .177(.000) & .000(.000) & .940(.007) & .194(.000) & .000(.000)\\
(7) N & .000(.001) & .000(.001) & .936(.008) & .152(.000) & .496(.009) & .936(.008) & .152(.000) & .494(.009)\\
(8) S & -.007(.002) & -.033(.002) & .903(.009) & .175(.000) & .000(.000) & .940(.007) & .212(.000) & .000(.000)\\
(8) N & -.001(.001) & .000(.001) & .917(.009) & .151(.000) & .496(.010) & .920(.009) & .154(.000) & .504(.010)\\
(9) S & -.169(.005) & -.264(.005) & .290(.026) & .242(.001) & .000(.000) & .322(.027) & .268(.001) & .000(.000)\\
(9) N & .000(.002) & .000(.001) & .984(.006) & .218(.001) & .625(.014) & .992(.004) & .251(.001) & .663(.014)\\
(10) S & -.078(.002) & -.155(.002) & .415(.016) & .138(.000) & .000(.000) & .495(.016) & .159(.000) & .000(.000)\\
(10) N & .000(.001) & .000(.000) & .976(.004) & .120(.000) & .609(.008) & .992(.002) & .149(.000) & .668(.007)\\
(11) S & -.063(.002) & -.150(.002) & .553(.016) & .143(.000) & .000(.000) & .612(.015) & .149(.000) & .000(.000)\\
(11) N & .000(.001) & .000(.000) & .977(.004) & .120(.000) & .586(.008) & .988(.003) & .136(.000) & .621(.008)\\
(12) S & -.081(.002) & -.154(.002) & .413(.016) & .141(.000) & .000(.000) & .485(.016) & .158(.000) & .000(.000)\\
(12) N & .000(.001) & .000(.000) & .976(.005) & .120(.000) & .608(.008) & .991(.002) & .147(.000) & .665(.007)\\
(13) S & -.034(.002) & -.122(.002) & .848(.011) & .178(.000) & .000(.000) & .895(.010) & .198(.000) & .000(.000)\\
(13) N & .000(.001) & .000(.000) & .985(.003) & .150(.000) & .593(.008) & .985(.003) & .150(.000) & .593(.008)\\
(14) S & -.052(.002) & -.126(.002) & .745(.014) & .177(.000) & .000(.000) & .852(.011) & .219(.000) & .000(.000)\\
(14) N & .000(.001) & .000(.000) & .988(.003) & .149(.000) & .624(.008) & .989(.003) & .151(.000) & .628(.008)\\
(15) S & -.028(.002) & -.122(.002) & .863(.011) & .180(.000) & .000(.000) & .897(.009) & .198(.000) & .000(.000)\\
(15) N & .000(.001) & .000(.000) & .985(.003) & .151(.000) & .593(.008) & .985(.003) & .151(.000) & .593(.008)\\
(16) S & -.046(.002) & -.126(.002) & .772(.013) & .178(.000) & .000(.000) & .845(.011) & .219(.000) & .000(.000)\\
(16) N & .000(.001) & .000(.000) & .987(.003) & .149(.000) & .624(.008) & .988(.003) & .151(.000) & .628(.008)\\
\hline
\end{tabular}}
\caption{Simulation settings (1)-(16).  S and N rows are for results for signal and noise variables respectively.}\label{tab-3-asy}
\end{table}
\end{landscape}

It is reassuring to see that, in all cases, the confidence intervals for the noise variables have close to nominal coverage and the $p$-values for the noise variables appear to be uniformly distributed.  Thus, our methodology is providing a reliable method for identifying signal variables, with uncertainty quantification.  On the other hand, while the confidence intervals for the signal variables have good coverage when $p=10$ (particularly with $\widetilde{\boldsymbol{\Theta}}$), it is much more challenging ensure nominal adequate for the signal variables in the $p=300$ case.  Apparently, the sample size needs to be very large for the asymptotics to `kick in', to the extent that we can think, for instance, that \textbf{(A4)(c)} is satisfied.  The greater width of the intervals when using $\widetilde{\boldsymbol{\Theta}}$ yields improved coverage for the signal variables, but leads to some over-coverage for the noise variables.

One approach in high-dimensional settings, then, is to use our methodology as a screening method to identify signal variables (with false discovery guarantees), and then use the standard MPLE inference to obtain confidence intervals for the signal variables at a second stage.

\subsection{Real data analysis}\label{sec-realdata}	

In this section, we apply our method to a diffuse large B-cell lymphoma (DLBCL) dataset, comprising survival times of 240 DLBCL patients and gene expression data from 7399 genes \citep{RosenwaldEtal2001}.  
To reduce dimensionality, we computed the Lasso path, noting that the cross-validation algorithm picked the 16th largest value of $\lambda$ on our grid of size 100.  In total, 84 variables were selected at some stage in the first 25 $\lambda$ values, and we therefore retained these 84 variables in our subsequent analysis.


In \Cref{fig-2}, we plot the glmnet solution paths, with solid and black paths being the ones for those variables deemed to be significant according to our methodology, and dashed and grey paths for those variables deemed insignificant.  The left and right panels correspond to the use of $\widehat{\boldsymbol{\Theta}}$ and $\widetilde{\boldsymbol{\Theta}}$ respectively, and the red vertical lines indicate the regularisation parameter values chosen by cross-validation.  The only difference between the inferences drawn from the two precision matrix estimates is the confidence interval widths, so the selected variables when using $\widehat{\boldsymbol{\Theta}}$ are a proper subset of those obtained using $\widetilde{\boldsymbol{\Theta}}$.   

It can be seen that some variables enter the model fairly early along the path, but appear not significant according to our methods.  These variables are often omitted from the model at a later stage along the path, as other variables enter.  This observation is demonstrated in \Cref{tab-real-asy}, which presents the median life-spans of the corresponding variables, where the life-span is defined as the proportion of the locations on the solution paths for which a certain variable is chosen.  

\begin{figure*}
\centering
\includegraphics[width = \textwidth]{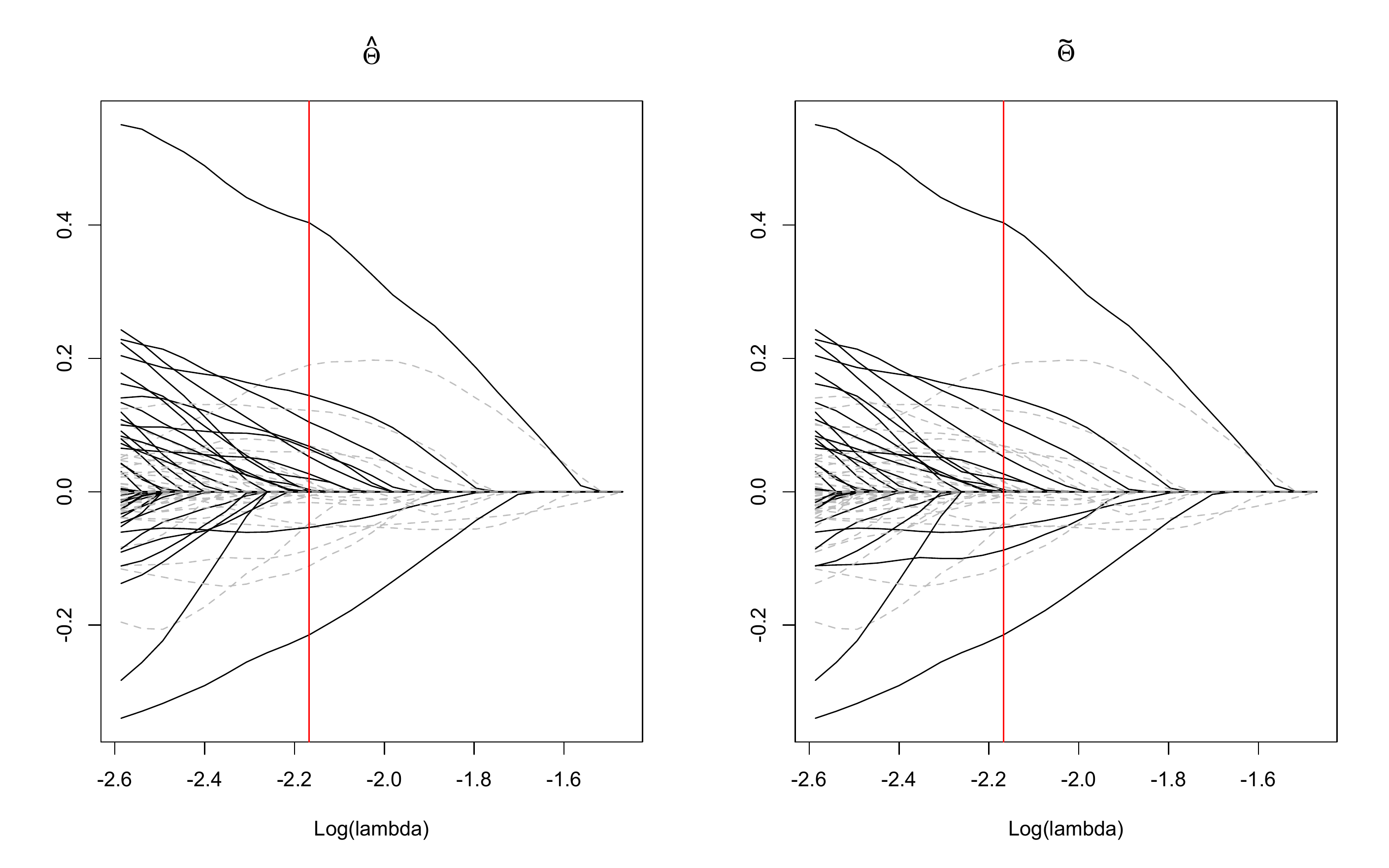}	
\caption{Solution paths}\label{fig-2}
\end{figure*}

\begin{table}[htbp]
\centering
\begin{tabular}{c c c | c c c}
\hline
\multicolumn{3}{c|}{$\widehat{\boldsymbol{\Theta}}$} &  \multicolumn{3}{c}{$\widetilde{\boldsymbol{\Theta}}$}  \\
No. & Significant & Insignificant & No. & Significant & Insignificant \\
\hline
41 & 0.78 & 0.26 & 32 & 0.78 & 0.35\\
\hline
\end{tabular}
\caption{Median life-spans for variables deemed significant and insignificant.}\label{tab-real-asy}
\end{table}

\section*{Appendix}

We first recall the following basic facts about the operator norms of a matrix.
\begin{lemma}\label{lem:multi}
	For any matrix $\boldsymbol{A} = (A_{ij}) \in\mathbb{R}^{m\times m}$, we have
	\[
	\|\boldsymbol{A}\|_{\mathrm{op},1} = \max_{j=1, \ldots, m}\sum_{i=1}^m |A_{ij}| \quad \text{and} \quad \|\boldsymbol{A}\|_{\mathrm{op},\infty} = \max_{i=1, \ldots, m}\sum_{j=1}^m |A_{ij}|
	\]
\end{lemma}
\begin{proof}
For any $\boldsymbol{v} = (v_1,\ldots,v_m)^\top \in \mathbb{R}^m \setminus \{\boldsymbol{0}\}$, we have
	\[
	\frac{\|\boldsymbol{A}\boldsymbol{v}\|_{1}}{\|\boldsymbol{v}\|_1} = \frac{\sum_{i=1}^m \bigl|\sum_{j=1}^m A_{ij}v_j\bigr|}{\sum_{j=1}^m |v_j|} \leq \frac{\sum_{j=1}^m |v_j| \sum_{i=1}^m |A_{ij}|}{\sum_{j=1}^m |v_j|} \leq \max_{j=1,\ldots,m} \sum_{i=1}^m |A_{ij}|.
	\] 
On the other hand, suppose initially that $\boldsymbol{A} \neq 0$.  Let $j^* \in \argmax_{j = 1, \ldots, m}\sum_{i=1}^m|A_{ij}|$, and let $\boldsymbol{v} = (v_1,\ldots,v_m)^\top \in \mathbb{R}^m$ be given by $v_i = \mathrm{sgn}(A_{ij^*})\mathbbm{1}_{\{i = j^*\}}$ for $i=1,\ldots,m$.  Then
	\[
	\max_{j = 1, \ldots, m}\sum_{i=1}^m |A_{ij}| = \sum_{i=1}^m |A_{ij^*}| = \frac{\|\boldsymbol{A}\boldsymbol{v}\|_{1}}{\|\boldsymbol{v}\|_1}.
	\]
This equality also holds for any $\boldsymbol{v} \neq 0$ when $\boldsymbol{A} = 0$, so the result for the case of $\|\cdot\|_{\mathrm{op,1}}$ follows.  The argument for the case of $\|\cdot\|_{\mathrm{op,\infty}}$ is similar and is omitted.
\end{proof}

The lemma below is key in deriving both the asymptotic normality of the leading term and the asymptotic negligibility of the residual terms.  It provides conditions under which we can control the deviations of the process $\{\bar{\boldsymbol{Z}}(t,\boldsymbol{\beta}):t \in [0,t_*)\}$ from $\{\boldsymbol{\mu}(t,\boldsymbol{\beta}):t \in [0,t_*)\}$, where $t_* \in [0,t_+)$ is an (initially arbitrary) truncation.  
It will be convenient to define some notation.  Let 
\begin{align}
\label{Eq:Lmu}
L_{\boldsymbol{\mu}}(\|\boldsymbol{\beta}\|_1) &:= \frac{e^{2\|\boldsymbol{\beta}\|_1K_Z}(L + K_Z\|f_T\|_\infty + 2K_Z\|\boldsymbol{\beta}\|_1L) + K_Ze^{2\|\boldsymbol{\beta}\|_1K_Z}(\|f_T\|_\infty + 2\|\boldsymbol{\beta}\|_1L)}{\bar{F}_T(t_*)}, \\
\label{Eq:LZ}
L_{\bar{\boldsymbol{Z}}}(\|\boldsymbol{\beta}\|_1) &:= L + 4K_Z\|\boldsymbol{\beta}\|_1L. 
\end{align}
\begin{lemma}
\label{Lemma:Zbarmu}
Assume~\textbf{(A1)} and~\textbf{(A2)(a)} and let $t_* \in \mathcal{T}$.  For $\boldsymbol{\beta} \in \mathbb{R}^p$, let $\epsilon_n = \epsilon_n(\|\boldsymbol{\beta}\|_1) := \frac{3K_Ze^{2\|\boldsymbol{\beta}\|_1K_Z}}{n\bar{F}_T(t_*) - n^{1/2}(\log n)\bar{F}_T^{1/2}(t_*)}$ and fix $\epsilon > \epsilon_n$.  Set
\[
h_0 = h_0(n,\epsilon) := \min\biggl\{\frac{1}{2(\|\boldsymbol{\beta}\|_1+1)L} \, , \, \frac{\epsilon - \epsilon_n}{3L_{\bar{\boldsymbol{Z}}}(\|\boldsymbol{\beta}\|_1)} \, , \, \frac{\epsilon}{3L_{\boldsymbol{\mu}}(\|\boldsymbol{\beta}\|_1)} \, , \, \frac{\bar{F}_T(t_*)}{2\|f_T\|_\infty}\biggr\}.
\]
Then, writing $M_0 = M_0(n,\epsilon) := t_*/h_0 + n^3 + 1$, we have
\begin{align*}
\mathbb{P}\biggl(\sup_{t \in [0,t_*)} &\|\bar{\boldsymbol{Z}}(t,\boldsymbol{\beta}) - \boldsymbol{\mu}(t,\boldsymbol{\beta})\|_\infty > \epsilon\biggr) \leq \frac{1}{2n} + e^{-(\log^2 n)/2} \\
&+ (2p+2)M_0\exp\biggl(-\frac{n\epsilon^2\bar{F}_T(t_*)}{1152K_Ze^{4\|\boldsymbol{\beta}\|_1K_Z}(K_Z + \frac{\epsilon}{36})}\biggr) + 2M_0\exp\biggl(-\frac{3n\bar{F}_T(t_*)}{28e^{4\|\boldsymbol{\beta}\|_1K_Z)}}\biggr).
\end{align*}
\end{lemma}
\noindent \textbf{Remark}: Consider an asymptotic regime in which \textbf{(A3)(a)} holds and $\bar{F}(t_*) = O\bigl(n^{-(1/2 -\delta)}\bigr)$ for some $\delta \in (0,1/2)$.  If $\int_0^{t_+} t^\alpha f_T(t) \, dt < \infty$ for some $\alpha > 0$ (i.e.\ \textbf{(A2)(b)} holds), then for such $t_*$,
\[
n^{-(1/2-\delta)} = \int_{F_T^{-1}(1 - n^{-(1/2 -\delta)})}^{t_+} f_T \leq \frac{1}{\bigl\{F_T^{-1}(1 - n^{-(1/2 -\delta)})\bigr\}^\alpha} \int_0^{t_+} t^\alpha f_T(t) \, dt;
\]
in other words, $F_T^{-1}(1 - n^{-(1/2 -\delta)}) = O(n^{(1/2-\delta)/\alpha})$.  It therefore follows from Lemma~\ref{Lemma:Zbarmu} that if \textbf{(A1)},~\textbf{(A2)} and~\textbf{(A3)(a)} hold, then
\[
\sup_{t \in [0,t_*)} \|\bar{\boldsymbol{Z}}(t,\boldsymbol{\beta}) - \boldsymbol{\mu}(t,\boldsymbol{\beta})\|_\infty \stackrel{p}{\rightarrow} 0.
\]
\begin{proof}
As a first step, we prove that the process $\{\boldsymbol{\mu}(t,\boldsymbol{\beta}):t \in [0,t_*)\}$ inherits a Lipschitz property from $\{\boldsymbol{Z}(t,\boldsymbol{\beta}):t \in [0,t_*)\}$.  In fact, writing $\tilde{w}(t,\boldsymbol{\beta}) := Y(t)e^{\boldsymbol{\beta}^\top \boldsymbol{Z}(t)}$, for $t,t+h \in [0,t_*)$ with $h \in (0,h_0]$ (so in particular $(\|\boldsymbol{\beta}\|_1+1)Lh \leq 1/2$), 
\begin{align*}
\bigl|\mathbb{E}\bigl\{&\tilde{w}(t+h,\boldsymbol{\beta})\bigr\} - \mathbb{E}\bigl\{\tilde{w}(t,\boldsymbol{\beta})\bigr\}\bigr| \\
&\leq \bigl|\mathbb{E}\bigl\{\tilde{w}(t+h,\boldsymbol{\beta})\bigr\} - \mathbb{E}\bigl\{Y(t)e^{\boldsymbol{\beta}^{\top}\boldsymbol{Z}(t+h)}\bigr\}\bigr| + \bigl|\mathbb{E}\bigl\{Y(t)e^{\boldsymbol{\beta}^{\top}\boldsymbol{Z}(t+h)}\bigr\} - \mathbb{E}\bigl\{\tilde{w}(t,\boldsymbol{\beta})\bigr\}\bigr| \\
&\leq e^{\|\boldsymbol{\beta}\|_1K_Z}\|f_T\|_\infty h + e^{\|\boldsymbol{\beta}\|_1K_Z}(e^{\|\boldsymbol{\beta}\|_1Lh} - 1) \leq e^{\|\boldsymbol{\beta}\|_1K_Z}(\|f_T\|_\infty + 2\|\boldsymbol{\beta}\|_1L)h.
\end{align*}
Similarly, again for $t, t+h \in [0,t_*)$ and $h \in (0,h_0]$, 
\begin{align*}
\bigl\|\mathbb{E}\bigl\{\boldsymbol{Z}(t+h)&\tilde{w}(t+h,\boldsymbol{\beta})\bigr\} - \mathbb{E}\bigl\{\boldsymbol{Z}(t)\tilde{w}(t,\boldsymbol{\beta})\bigr\}\bigr\|_\infty \\
&\leq e^{\|\boldsymbol{\beta}\|_1K_Z}Lh + K_Ze^{\|\boldsymbol{\beta}\|_1K_Z}\|f_T\|_\infty h + K_Ze^{\|\boldsymbol{\beta}\|_1K_Z}(e^{\|\boldsymbol{\beta}\|_1Lh} - 1) \\
&\leq e^{\|\boldsymbol{\beta}\|_1K_Z}(L + K_Z\|f_T\|_\infty + 2K_Z\|\boldsymbol{\beta}\|_1L)h.
\end{align*}
It follows that provided $h \in (0,h_0]$, so that $\bar{F}_T(t+h) \geq \bar{F}_T(t)/2$ for $t,t+h \in [0,t_*)$, we have
\begin{align}
\label{Eq:murjplush}
\|\boldsymbol{\mu}&(t+h,\boldsymbol{\beta}) - \boldsymbol{\mu}(t,\boldsymbol{\beta})\|_\infty \nonumber \\
&= \biggl\|\frac{\mathbb{E}\bigl\{\boldsymbol{Z}(t+h)\tilde{w}(t+h,\boldsymbol{\beta})\bigr\}\mathbb{E}\bigl\{\tilde{w}(t,\boldsymbol{\beta})\bigr\} - \mathbb{E}\bigl\{\boldsymbol{Z}(t)\tilde{w}(t,\boldsymbol{\beta})\bigr\}\mathbb{E}\bigl\{\tilde{w}(t+h,\boldsymbol{\beta})\bigr\}}{\mathbb{E}\bigl\{\tilde{w}(t+h,\boldsymbol{\beta})\bigr\}\mathbb{E}\bigl\{\tilde{w}(t,\boldsymbol{\beta})\bigr\}}\biggr\|_\infty \nonumber \\
&\leq L_{\boldsymbol{\mu}}(\|\boldsymbol{\beta}\|_1)h,
\end{align}
where $L_{\boldsymbol{\mu}}(\|\boldsymbol{\beta}\|_1)$ was defined in~\eqref{Eq:Lmu}.  We now aim to prove a similar property for the process $\{\bar{\boldsymbol{Z}}(t,\boldsymbol{\beta}):t \in [0,t_*)\}$ (though this process may have jumps). Let $M := n^3$, and let $s_j := F_T^{-1}(j/M)$ for $j=0,1,\ldots,M-1$.  Let $E_j := \sum_{i=1}^n \mathbbm{1}_{\{T_i \in [s_j,s_{j+1})\}}$, and let $\Omega_0 := \cap_{j=1}^M \{E_j \leq 1\}$.  Then
\[
\mathbb{P}(\Omega_0^c) \leq M\biggl\{1 - \biggl(1-\frac{1}{M}\biggr)^n - \frac{n}{M}\biggl(1-\frac{1}{M}\biggr)^{n-1}\biggr\} \leq \frac{1}{2n}.
\]
Now, fix $t \in [0,t_*)$ and $h \in (0,h_0]$ such that $t,t+h \in [s_j,s_{j+1})$ for some $j$, and let $R_t := \{i:Y_i(t) = 1\}$ denote the risk set at time $t$.  If $\sum_{i \in R_t} \mathbbm{1}_{\{T_i \in [t_i,t_i+h)\}} = 0$, i.e.\ there are no observed events in $[t,t+h)$, then
\begin{align*}
\sum_{i \in R_t} |w_i(t+h,\boldsymbol{\beta}) - w_i(t,\boldsymbol{\beta})| &= \sum_{i \in R_t} \biggl|\frac{e^{\boldsymbol{\beta}^\top\boldsymbol{Z}_i(t+h)}\sum_{j \in R_t} e^{\boldsymbol{\beta}^\top\boldsymbol{Z}_j(t)} - e^{\boldsymbol{\beta}^\top\boldsymbol{Z}_i(t)}\sum_{j \in R_t} e^{\boldsymbol{\beta}^\top\boldsymbol{Z}_j(t+h)}}{\bigl(\sum_{j \in R_t} e^{\boldsymbol{\beta}^\top\boldsymbol{Z}_j(t+h)}\bigr)\bigl(\sum_{j \in R_t} e^{\boldsymbol{\beta}^\top\boldsymbol{Z}_j(t)}\bigr)}\biggr| \\
&\leq 2(e^{\|\boldsymbol{\beta}\|_1Lh}-1).
\end{align*}
On the other hand, if there is one observed event (corresponding to the individual $i^*$) in $[t,t+h)$, then
\begin{align*}
\sum_{i \in R_t} |w_i(t+h,\boldsymbol{\beta}) - w_i(t,\boldsymbol{\beta})| &= \sum_{i \in R_t \setminus \{i^*\}} |w_i(t+h,\boldsymbol{\beta}) - w_i(t,\boldsymbol{\beta})| + w_{i^*}(t,\boldsymbol{\beta}) \\
&\leq 2(e^{\|\boldsymbol{\beta}\|_1Lh}-1) + \frac{e^{2\|\boldsymbol{\beta}\|_1K_Z}}{|R_t|}.
\end{align*}
It follows that on the event $\Omega_0$, if $t \in [0,t_*)$ and $h \in (0,h_0]$ are such that $t,t+h \in [s_j,s_{j+1})$ for some $j$, then
\begin{align*}
\|\bar{\boldsymbol{Z}}(t+h,\boldsymbol{\beta}) - \bar{\boldsymbol{Z}}(t,\boldsymbol{\beta})\|_\infty &= \biggl\|\sum_{i \in R_t}  \boldsymbol{Z}_i(t+h)w_i(t+h,\boldsymbol{\beta}) - \sum_{i \in R_t} \boldsymbol{Z}_i(t)w_i(t,\boldsymbol{\beta})\biggr\|_\infty \\
&\leq L_{\bar{\boldsymbol{Z}}}(\|\boldsymbol{\beta}\|_1)h + \frac{K_Ze^{2\|\boldsymbol{\beta}\|_1K_Z}}{|R_t|},
\end{align*}
where $L_{\bar{\boldsymbol{Z}}}(\|\boldsymbol{\beta}\|_1)$ was defined in~\eqref{Eq:LZ}.  

Now let $\Omega_1 := \bigl\{|R_{t_*}| \geq n\bar{F}_T(t_*) - n^{1/2}(\log n)\bar{F}_T^{1/2}(t_*)\bigr\}$, so that by a standard Binomial tail bound \citep[e.g.][p.~440]{ShorackWellner1986},
\[
\mathbb{P}(\Omega_1^c) \leq e^{-(\log^2 n)/2}.
\]
Fix $\epsilon > \epsilon_n$, and partition $[0,t_*)$ into at most $\lceil t_*/h_0 \rceil + M \leq M_0$ intervals $\{[r_j,r_{j+1}):j=0,\ldots,M_0-1\}$ such that for each $j$, there exists $k$ for which $[r_j,r_{j+1}) \subseteq [s_k,s_{k+1})$, and such that $|r_{j+1} - r_j| \leq h_0$.  Then
\begin{align}
\label{Eq:Zrjplush}
\mathbb{P}\biggl(\sup_{h \in (0,h_0]} \max_{j=0,1,\ldots,M_0-1} \|\bar{\boldsymbol{Z}}(r_j+h,\boldsymbol{\beta}) - \bar{\boldsymbol{Z}}(r_j,\boldsymbol{\beta})\|_\infty > \frac{\epsilon}{3}\biggr) &\leq \mathbb{P}(\Omega_0^c) + \mathbb{P}(\Omega_1^c) \nonumber \\
&\leq \frac{1}{2n} + e^{-(\log^2 n)/2}.
\end{align}
Finally, we seek to control the difference between $\bar{\boldsymbol{Z}}(\cdot,\boldsymbol{\beta})$ and $\boldsymbol{\mu}(\cdot,\boldsymbol{\beta})$ at $r_0,\ldots,r_{M_0}$.  To this end, note that for any $t \in [0,t_*)$,
\begin{align}
\label{Eq:rj2}
	 \bigl\|\bar{\boldsymbol{Z}}(t, \boldsymbol{\beta}) - \boldsymbol{\mu}(t, \boldsymbol{\beta})\bigr\|_\infty &= \biggl\|\frac{n^{-1}\sum_{i=1}^n \boldsymbol{Z}_i(t)\tilde{w}_i(t, \boldsymbol{\beta})}{n^{-1}\sum_{j=1}^n\tilde{w}_j(t, \boldsymbol{\beta})} - \frac{\mathbb{E}\bigl\{\boldsymbol{Z}(t)\tilde{w}(t, \boldsymbol{\beta})\bigr\}}{\mathbb{E}\bigl\{\tilde{w}(t, \boldsymbol{\beta})\bigr\}}\biggr\|_\infty \nonumber \\
	&\leq \frac{1}{n^{-1}\sum_{j=1}^n\tilde{w}_j(t, \boldsymbol{\beta})}\biggl\|\frac{1}{n}\sum_{i=1}^n \boldsymbol{Z}_i(t)\tilde{w}_i(t, \boldsymbol{\beta}) - \mathbb{E}\bigl\{\boldsymbol{Z}(t)\tilde{w}(t,\boldsymbol{\beta})\bigr\}\biggr\|_\infty \nonumber \\
	&\hspace{0.5cm}+ K_Z\mathbb{E}\bigl\{\tilde{w}(t, \boldsymbol{\beta})\bigr\}\biggl|\frac{1}{n^{-1}\sum_{j=1}^n \tilde{w}_j(t,\boldsymbol{\beta})} - \frac{1}{\mathbb{E}\bigl\{\tilde{w}(t, \boldsymbol{\beta})\bigr\}}\biggr|. 
\end{align}
Let  
\[
\Omega_2 := \biggl\{\frac{1}{n}\sum_{j=1}^n\tilde{w}_j(t, \boldsymbol{\beta}) \geq \frac{1}{2}\mathbb{E}\bigl\{\tilde{w}(t, \boldsymbol{\beta})\bigr\} \biggr\},
\]
so that by Bernstein's inequality,
\begin{equation}
\label{Eq:Omega2}
\mathbb{P}(\Omega_2^c) \leq \exp\biggl(-\frac{n\mathbb{E}^2\tilde{w}(t, \boldsymbol{\beta})}{8\bigl\{\mathbb{E}\tilde{w}^2(t, \boldsymbol{\beta}) + e^{\|\boldsymbol{\beta}\|_1K_Z}\mathbb{E}\tilde{w}(t, \boldsymbol{\beta})/6\bigr\}}\biggr) \leq \exp\biggl(-\frac{3n\bar{F}_T(t)}{28\exp(4\|\boldsymbol{\beta}\|_1K_Z)}\biggr).
\end{equation}
Then, by Bernstein's inequality again,
\begin{align}
\label{Eq:FirstTerm}
\mathbb{P}\biggl(\biggl\{&\frac{1}{n^{-1}\sum_{j=1}^n\tilde{w}_j(t, \boldsymbol{\beta})}\biggl\|\frac{1}{n}\sum_{i=1}^n \boldsymbol{Z}_i(t)\tilde{w}_i(t, \boldsymbol{\beta}) - \mathbb{E}\bigl\{\boldsymbol{Z}(t)\tilde{w}(t,\boldsymbol{\beta})\bigr\}\biggr\|_\infty > \frac{\epsilon}{2}\biggr\} \ \bigcap \ \Omega_2\biggr) \nonumber \\
&\leq \mathbb{P}\biggl(\biggl\|\frac{1}{n}\sum_{i=1}^n \boldsymbol{Z}_i(t)\tilde{w}_i(t, \boldsymbol{\beta}) - \mathbb{E}\bigl\{\boldsymbol{Z}(t)\tilde{w}(t,\boldsymbol{\beta})\bigr\}\biggr\|_\infty > \frac{\epsilon \mathbb{E}\bigl\{\tilde{w}(t,\boldsymbol{\beta})\bigr\}}{4}\biggr) \nonumber \\
&\leq 2p\exp\biggl(-\frac{n\epsilon^2\bar{F}_T(t)}{32K_Z\exp\bigl(4\|\boldsymbol{\beta}\|_1K_Z\bigr)(K_Z + \epsilon/12)}\biggr).
\end{align}
Now the mean value theorem, for $x,y > 0$ with $x \geq y/2$,
\begin{align*}
\biggl|\frac{1}{x} - \frac{1}{y}\biggr| \leq \sup_{x_* = (1-t)x+ty:t \in [0,1]} \frac{|x - y|}{x_*^2} \leq \frac{4|x-y|}{y^2}.
\end{align*}
Hence, by another application of Bernstein's inequality,
\begin{align}
\label{Eq:SecondTerm}
\mathbb{P}\biggl(\biggl\{K_Z\mathbb{E}\bigl\{\tilde{w}(t, \boldsymbol{\beta})\bigr\}&\biggl|\frac{1}{n^{-1}\sum_{i=1}^n \tilde{w}_i(t,\boldsymbol{\beta})} - \frac{1}{\mathbb{E}\bigl\{\tilde{w}(t, \boldsymbol{\beta})\bigr\}}\biggr|  > \frac{\epsilon}{2}\biggr\} \ \bigcap \ \Omega_2\biggr) \nonumber \\
&\leq \mathbb{P}\biggl(\biggl|\frac{1}{n}\sum_{i=1}^n \tilde{w}_i(t, \boldsymbol{\beta}) - \mathbb{E}\bigl\{\tilde{w}(t, \boldsymbol{\beta})\bigr\}\biggr| > \frac{\epsilon \mathbb{E}\bigl\{\tilde{w}(t,\boldsymbol{\beta})\bigr\}}{8K_Z}\biggr)   \nonumber\\
&\leq 2\exp\biggl(-\frac{n\epsilon^2\bar{F}_T(t)}{128K_Z\exp(4\|\boldsymbol{\beta}\|_1K_Z)(K_Z+ \epsilon/24)}\biggr).
\end{align}
It follows from~\eqref{Eq:rj2},~\eqref{Eq:Omega2},~\eqref{Eq:FirstTerm} and~\eqref{Eq:SecondTerm} that for any $\epsilon > 0$ and $t \in [0,t_*)$,
\begin{align}
\label{Eq:ZbarmuDiff}
\mathbb{P}\bigl(\bigl\|\bar{\boldsymbol{Z}}(t, \boldsymbol{\beta}) - \boldsymbol{\mu}(t, \boldsymbol{\beta})\bigr\|_\infty > \epsilon\bigr) &\leq (2p+2)\exp\biggl(-\frac{n\epsilon^2\bar{F}_T(t)}{128K_Z\exp\bigl(4\|\boldsymbol{\beta}\|_1K_Z\bigr)(K_Z + \epsilon/12)}\biggr) \nonumber \\
&\hspace{4cm}+ \exp\biggl(-\frac{3n\bar{F}_T(t)}{28\exp(4\|\boldsymbol{\beta}\|_1K_Z)}\biggr).
\end{align}
From~\eqref{Eq:murjplush},~\eqref{Eq:Zrjplush} and~\eqref{Eq:ZbarmuDiff}, together with the fact that $L_{\boldsymbol{\mu}}(\|\boldsymbol{\beta}\|_1)h_0 \leq \epsilon/3$, we conclude that for $\epsilon > \epsilon_n$,
\begin{align*}
\mathbb{P}&\biggl(\sup_{t \in [0,t_*)} \|\bar{\boldsymbol{Z}}(t,\boldsymbol{\beta}) - \boldsymbol{\mu}(t,\boldsymbol{\beta})\|_\infty > \epsilon\biggr) \\
&\leq \mathbb{P}\biggl(\sup_{h \in (0,h_0]} \max_{j = 0,1,\ldots,M_0-1} \|\bar{\boldsymbol{Z}}(r_j+h,\boldsymbol{\beta}) - \bar{\boldsymbol{Z}}(r_j,\boldsymbol{\beta})\|_\infty > \frac{\epsilon}{3}\biggr) \\
&\hspace{4cm}+ \mathbb{P}\biggl(\max_{j = 0,1,\ldots,M_0-1} \|\bar{\boldsymbol{Z}}(r_j,\boldsymbol{\beta}) - \boldsymbol{\mu}(r_j,\boldsymbol{\beta})\|_\infty > \frac{\epsilon}{3}\biggr) \\
&\leq \frac{1}{2n} + e^{-(\log^2 n)/2} + (2p+2)M_0\exp\biggl(-\frac{n\epsilon^2\bar{F}_T(t_*)}{1152K_Z\exp\bigl(4\|\boldsymbol{\beta}\|_1K_Z\bigr)(K_Z + \epsilon/36)}\biggr) \\
&\hspace{4cm}+ M_0\exp\biggl(-\frac{3n\bar{F}_T(t_*)}{28\exp(4\|\boldsymbol{\beta}\|_1K_Z)}\biggr),
\end{align*}
as required.
\end{proof}  

The lemma below is used several times in controlling the residual terms in~\eqref{Eq:1}.
\begin{lemma}\label{lem-aclime-lem1} 
Assume~\textbf{(A1)},~\textbf{(A2)(a)},~\textbf{(A3)(b)} and \textbf{(A4)(a)}.  For $\widehat{\boldsymbol{\beta}}$ in~\eqref{eq-betahat}, let $\lambda \asymp n^{-1/2}\log^{1/2} (np)$.  Then, for every $\eta \in (0,1/3)$, 
	\[
	\|\widehat{\mathcal{V}}(\widehat{\boldsymbol{\beta}})\boldsymbol{\Sigma}^{-1} - \boldsymbol{I}\|_{\infty} = O_p\biggl\{\max\biggl(\|\boldsymbol{\Sigma}^{-1}\|_{\mathrm{op},1}\frac{d_o\log^{1/2}(np)}{n^{1/2}} \, , \, \|\boldsymbol{\Sigma}^{-1}\|_{\mathrm{op},1}n^{-(1/3-\eta)}\biggr)\biggr\}.
	\]
\end{lemma}

\begin{proof}
Writing $W_1 := n^{-1}\sum_{i=1}^n\int_{\mathcal{T}} \{\boldsymbol{Z}_i(s) - \boldsymbol{\mu}(s,\boldsymbol{\beta}^o)\}^{\otimes 2}\, dN_i(s) - \boldsymbol{\Sigma}$, we have
		\begin{align}
\label{Eq:LongDisplay}
		\|\widehat{\mathcal{V}}(\widehat{\boldsymbol{\beta}}) &- \boldsymbol{\Sigma}\|_{\infty} \nonumber \\
&\leq  \biggl\|\frac{1}{n}\sum_{i=1}^n\int_{\mathcal{T}} \bigl\{(\boldsymbol{Z}_i(s) - \bar{\boldsymbol{Z}}(s,\widehat{\boldsymbol{\beta}}))^{\otimes 2} - (\boldsymbol{Z}_i(s) - \boldsymbol{\mu}(s,\boldsymbol{\beta}^o))^{\otimes 2}\bigr\}\, dN_i(s)\biggr\|_{\infty} + \|W_1\|_\infty \nonumber \\
		&\leq \biggl\|\frac{1}{n}\sum_{i=1}^n \int_{\mathcal{T}}\bigl(\boldsymbol{Z}_i(s) - \bar{\boldsymbol{Z}}(s,\widehat{\boldsymbol{\beta}})\bigr)\bigl(\bar{\boldsymbol{Z}}(s,\widehat{\boldsymbol{\beta}}) - \boldsymbol{\mu}(s,\boldsymbol{\beta}^o)\bigr)^{\top}\,dN_i(s)\biggr\|_{\infty} \nonumber \\
		&\hspace{1.5cm}+ \biggl\|\frac{1}{n}\sum_{i=1}^n \int_{\mathcal{T}}\bigl(\boldsymbol{Z}_i(s) - \boldsymbol{\mu}(s,\boldsymbol{\beta}^o)\bigr)\bigl(\bar{\boldsymbol{Z}}(s,\widehat{\boldsymbol{\beta}}) - \boldsymbol{\mu}(s,\boldsymbol{\beta}^o)\bigr)^{\top}\,dN_i(s)\biggr\|_{\infty} + \|W_1\|_\infty \nonumber \\
		&\leq \frac{4K_Z}{n} \sum_{i=1}^n \int_{\mathcal{T}}\|\bar{\boldsymbol{Z}}(s,\widehat{\boldsymbol{\beta}}) - \boldsymbol{\mu}(s,\boldsymbol{\beta}^o)\|_{\infty}\,dN_i(s) + \|W_1\|_\infty \nonumber \\
		&\leq  \frac{4K_Z}{n} \sum_{i=1}^n \int_{\mathcal{T}}\|\bar{\boldsymbol{Z}}(s,\widehat{\boldsymbol{\beta}}) - \bar{\boldsymbol{Z}}(s,\boldsymbol{\beta}^o)\|_{\infty}\,dN_i(s) \nonumber \\
&\hspace{3cm}+ \frac{4K_Z}{n} \sum_{i=1}^n \int_{\mathcal{T}}\|\bar{\boldsymbol{Z}}(s,\boldsymbol{\beta}^o) - \boldsymbol{\mu}(s,\boldsymbol{\beta}^o)\|_{\infty}\,dN_i(s)+ \|W_1\|_\infty.
		\end{align}
Now, for any $s \in \mathcal{T}$,
\begin{align}
\label{Eq:ZbarDiff}
\|\bar{\boldsymbol{Z}}(s,\widehat{\boldsymbol{\beta}}) - \bar{\boldsymbol{Z}}(s,\boldsymbol{\beta}^o)\|_{\infty} &= \biggl\|\frac{\sum_{i=1}^n \boldsymbol{Z}_i(s)Y_i(s)e^{\widehat{\boldsymbol{\beta}}^\top \boldsymbol{Z}_i(s)}}{\sum_{i=1}^n Y_i(s)e^{\widehat{\boldsymbol{\beta}}^\top \boldsymbol{Z}_i(s)}} - \frac{\sum_{i=1}^n \boldsymbol{Z}_i(s)Y_i(s)e^{\boldsymbol{\beta}^{o\top} \boldsymbol{Z}_i(s)}}{\sum_{i=1}^n Y_i(s)e^{\boldsymbol{\beta}^{o\top} \boldsymbol{Z}_i(s)}}\biggr\|_\infty \nonumber \\
&\leq 2K_Z(e^{K_Z\|\widehat{\boldsymbol{\beta}} - \boldsymbol{\beta}^o\|_1}-1).
\end{align}
It therefore follows from Lemma~\ref{Lem:lasso}(ii) that 
\begin{equation}
\label{Eq:LassoTerm}
\frac{4K_Z}{n} \sum_{i=1}^n \int_{\mathcal{T}}\|\bar{\boldsymbol{Z}}(s,\widehat{\boldsymbol{\beta}}) - \bar{\boldsymbol{Z}}(s,\boldsymbol{\beta}^o)\|_{\infty}\,dN_i(s) = O_p\biggl(\frac{d_o\log^{1/2}(np)}{n^{1/2}}\biggr).
\end{equation}
Next, fix an arbitrary $\eta \in (0,1/3)$ and let $t_* := F_T^{-1}(1-n^{-(1/3-\eta)})$.  Recalling that $R_t = \{i:Y_i(t) = 1\}$, set $\Omega_* := \bigl\{n\bar{F}(t_*)/2 \leq |R_{t_*}| \leq 3n\bar{F}(t_*)/2\bigr\}$.  Then by Hoeffding's inequality,
\begin{equation}
  \label{Eq:Omegastar}
\mathbb{P}(\Omega_*^c) \leq 2e^{-n^{1/3+2\eta}/2}.
\end{equation}
Moreover, on $\Omega_*$, 
\begin{align}
\label{Eq:BadTerm}
\frac{4K_Z}{n} \sum_{i=1}^n \int_{\mathcal{T}}\|\bar{\boldsymbol{Z}}(s,\boldsymbol{\beta}^o) - \boldsymbol{\mu}(s,\boldsymbol{\beta}^o)\|_{\infty}\,dN_i(s) &\leq 4K_Z\sup_{s \in [0,t_*)} \|\bar{\boldsymbol{Z}}(s,\boldsymbol{\beta}^o) - \boldsymbol{\mu}(s,\boldsymbol{\beta}^o)\|_{\infty} \nonumber \\
&\hspace{3cm}+ 12K_Z^2n^{-(1/3-\eta)} \nonumber \\
&= O_p(n^{-(1/3-\eta)}), 
\end{align}
where the final bound follows from Lemma~\ref{Lemma:Zbarmu}.  Finally, by \textbf{(A1)(a)} and Hoeffding's inequality, we have that for every $x > 0$,
	\[
	\mathbb{P}\bigl(\|W_1\|_{\infty} > x \bigr) \leq p(p+1)e^{-nx^2/(32K_Z^4)},
	\]
so that
\begin{equation}
\label{Eq:EasyTerm}
\|W_1\|_{\infty} = O_p\biggl(\frac{\log^{1/2} (np)}{n^{1/2}}\biggr).
\end{equation}
We conclude from~\eqref{Eq:LongDisplay},~\eqref{Eq:LassoTerm},~\eqref{Eq:Omegastar},~\eqref{Eq:BadTerm} and~\eqref{Eq:EasyTerm} that for every $\eta \in (0,1/3)$,
\begin{align*}
\|\widehat{\mathcal{V}}(\widehat{\boldsymbol{\beta}})\boldsymbol{\Sigma}^{-1} - \boldsymbol{I}\|_{\infty} & \leq \|\boldsymbol{\Sigma}^{-1}\|_{\mathrm{op},1}\|\widehat{\mathcal{V}}(\widehat{\boldsymbol{\beta}})- \boldsymbol{\Sigma}\|_{\infty} \\
&= O_p\biggl\{\max\biggl(\|\boldsymbol{\Sigma}^{-1}\|_{\mathrm{op},1}\frac{d_o\log^{1/2}(np)}{n^{1/2}} \, , \, \|\boldsymbol{\Sigma}^{-1}\|_{\mathrm{op},1}n^{-(1/3-\eta)}\biggr)\biggr\},
\end{align*}
as required.
\end{proof}
The following result is a consequence of~\citet[][Lemma~7.1]{CaiEtal2016}.
\begin{lemma}
\label{lem-aclime-lem-7}
Let $\boldsymbol{\Theta} = (\boldsymbol{\Theta}_1,\ldots,\boldsymbol{\Theta}_p)^\top = (\Theta_{ij}) \in \mathbb{R}^{p \times p}$ be symmetric and let $\check{\boldsymbol{\Theta}} = (\check{\boldsymbol{\Theta}}_1,\ldots,\check{\boldsymbol{\Theta}}_p)^\top$ be an estimator of $\boldsymbol{\Theta}$.  On the event
		\[
		\bigl\{\|\check{\boldsymbol{\Theta}}_j\|_1 \leq \|\boldsymbol{\Theta}_j\|_1, \, j = 1, \ldots, p\bigr\},
		\]
		we have
		\[
		\|\check{\boldsymbol{\Theta}} - \boldsymbol{\Theta}\|_{\mathrm{op},\infty} \leq 12\|\check{\boldsymbol{\Theta}} - \boldsymbol{\Theta}\|_{\infty} \max_{j=1,\ldots,p} \sum_{i=1}^p \mathbbm{1}_{\{\Theta_{ij} \neq 0\}}.
		\]
\end{lemma}


\textbf{Acknowledgements}: The third author is supported by Engineering and Physical Sciences Research Council fellowships EP/J017213/1 and EP/P031447/1 and grant RG81761 from the Leverhulme Trust.  The first and third authors would like to thank the Isaac Newton Institute for Mathematical Sciences for support and hospitality during the programme `Statistical Scalability' when work on this paper was undertaken. This work was supported by Engineering and Physical Sciences Research Council grant number EP/K032208/1.

\end{document}